\documentclass[journal,10pt,a4paper,final]{IEEEtran}
\usepackage{silence}
\usepackage{cite}

\usepackage{amsmath,amssymb,amsthm,fixmath}

\usepackage{color,colortbl}
\usepackage{xcolor}

\usepackage[inline]{enumitem}

\usepackage{siunitx}
\DeclareSIUnit{\dBm}{dBm}

\usepackage{graphicx}
\usepackage[labelformat=simple]{subcaption}

\usepackage{epstopdf}
\usepackage{cuted}
\usepackage{multirow,booktabs}
\usepackage{diagbox}
\usepackage{makecell}
\usepackage{hhline}
\usepackage{array} 
\newcolumntype{x}{!{\vrule width 2px}}
\newcolumntype{y}{!{\vrule width 1.5px}}

\usepackage{optidef}

\usepackage[hidelinks]{hyperref}

\usepackage[shortcuts,acronym,automake]{glossaries}
\makeglossaries
\newacronym{ai}{AI}{artificial intelligence}
\newacronym{awgn}{AWGN}{additive white Gaussian noise}
\newacronym{bcd}{BCD}{block coordinate descent}
\newacronym{bs}{BS}{base station}
\newacronym{cp}{CP}{control plane}
\newacronym{crc}{CRC}{cyclic redundancy check}
\newacronym{csi}{CSI}{channel state information}
\newacronym{csit}{CSIT}{channel state information at transmitter}
\newacronym{dft-s-ofdm}{DFT-s-OFDM}{Discrete Fourier Transform-spread-OFDM}
\newacronym{fbl}{FBL}{finite blocklength}
\newacronym{gan}{GAN}{generative adversarial network}
\newacronym{ibl}{IBL}{infinite blocklength}
\newacronym{iot}{IoT}{Internet of Things}
\newacronym{lfp}{LFP}{leakage-failure probability}
\newacronym{lp}{LP}{linear programming}
\newacronym{lut}{LUT}{look-up table}
\newacronym{mimo}{MIMO}{multi-input multi-output}
\newacronym{miso}{MISO}{multiple-input single-output}
\newacronym{ml}{ML}{machine learning}
\newacronym{mm}{MM}{Minorize-Maximization}
\newacronym{mmtc}{mMTC}{massive machine-type communication}
\newacronym{noma}{NOMA}{non-orthogonal multi-access}
\newacronym{nom}{NOM}{non-orthogonal multiplexing}
\newacronym{ofdm}{OFDM}{orthogonal frequency-division multiplexing}
\newacronym{ofdma}{OFDMA}{orthogonal frequency-division multiple access}
\newacronym{oma}{OMA}{orthogonal multiple access}
\newacronym{papr}{PAPR}{Peak-to-Average Power Ratio}
\newacronym{pdf}{PDF}{probability density function}
\newacronym{per}{PER}{packet error rate}
\newacronym{phy}{PHY}{physical}
\newacronym{pld}{PLD}{physical layer deception}
\newacronym{pls}{PLS}{physical layer security}
\newacronym{prb}{PRB}{physical resource block}
\newacronym{ran}{RAN}{radio access network}
\newacronym{sic}{SIC}{successive interference cancellation}
\newacronym{simo}{SIMO}{single-input multiple-output}
\newacronym{sinr}{SINR}{signal-to-interference-and-noise ratio}
\newacronym{snr}{SNR}{signal-to-noise ratio}
\newacronym{tdma}{TDMA}{time-division multiple access}
\newacronym{ue}{UE}{user equipment}
\newacronym{up}{UP}{user plane}
\newacronym{urllc}{URLLC}{ultra-reliable low-latency communication}

\usepackage{tcolorbox}
\tcbuselibrary{many}
\newtheorem{theorem}{Theorem}
\newtheorem{lemma}{Lemma}
\newtheorem{corollary}{Corollary}
\newtheorem{proposition}{Proposition}
\newtheorem{definition}{Definition}
\newtheorem{remark}{Remark}

\usepackage[norelsize,linesnumbered,ruled]{algorithm2e}
\SetKwRepeat{Do}{do}{while}
\makeatletter
\newcommand{\removelatexerror} {\let\@latex@error\@gobble}
\makeatother

\newcommand{\superscript}[1]{^{\mathrm{#1}}}
\newcommand{\subscript}[1]{_{\mathrm{#1}}}

\usepackage[normalem]{ulem}
\newcommand{\revise}[2]{{\color{red}\sout{#1}}{\color{blue}#2}} 
\renewcommand{\revise}[2]{#2} 

\usepackage[textsize=tiny,colorinlistoftodos]{todonotes}
\makeatletter
\define@key{todonotes}{bh}[]{
	\setkeys{todonotes}{author=\textbf{Bin}, color=lime!30}}%
\define@key{todonotes}{yz}[]{
	\setkeys{todonotes}{author=\textbf{Yao}, color=blue!30}}%
\define@key{todonotes}{wc}[]{
	\setkeys{todonotes}{author=\textbf{Wenwen}, color=green!30}}%
\makeatother

\newif\ifreviewmode
\reviewmodefalse

\ifreviewmode
\else
  \renewcommand{\todo}[1]{} 
  \renewcommand{\revise}[2]{#2} 
\fi

\newcommand\bob{\subscript{Bob}}
\newcommand\eve{\subscript{Eve}}

\begin{document}

\title{Physical Layer Deception as a Stackelberg Game:\\Strategy Regimes, Equilibrium, and Robust Design}

\author{Wenwen~Chen, 
    Bin~Han,~\IEEEmembership{Senior Member,~IEEE,}
    Yao~Zhu,~\IEEEmembership{Member,~IEEE,}
    Anke~Schmeink,~\IEEEmembership{Senior Member,~IEEE,} 
    Giuseppe~Caire,~\IEEEmembership{Fellow,~IEEE,}
    and
    Hans~D.~Schotten,~\IEEEmembership{Member,~IEEE}
    \thanks{W. Chen, B. Han, and H. D. Schotten are with RPTU University Kaiserslautern-Landau, Germany. Y. Zhu is with Wuhan University, China. A. Schmeink is with RWTH Aachen University, Germany. G. Caire is with Technical University of Berlin, Germany. H. D. Schotten is with the German Research Center for Artificial Intelligence (DFKI), Germany. 
    B. Han (bin.han@rptu.de) is the corresponding author.
    Anthropic's Claude 4.6 Opus assisted with language polishing, auxiliary derivation checking, and minor simulation debugging.
    }
}


\maketitle

\begin{abstract}
\Ac{pld} combines \ac{pls} with deception: the transmitter actively misleads the eavesdropper with falsified information. We model the transmitter--eavesdropper interaction as a Stackelberg game in which the transmitter commits to a resource allocation and encryption strategy, and each receiver best-responds by selecting among three decryption modes: Perception, Dropping, and Exclusion. Using semantic distortion as the metric, we derive closed-form switching surfaces that partition the parameter space into strategy regimes and identify conditions under which each regime dominates. The robust operating point, at the peak of the worst-case distortion envelope, is shown to be a Stackelberg equilibrium; iterative best-response dynamics oscillate around it with strictly lower time-averaged security. We evaluate the design under Nakagami-$m$ fading with static and adaptive transmitter strategies, benchmarked against a classical \ac{pls} baseline. Numerical results validate the regime characterization and show $12$--$55\%$ higher eavesdropper distortion than the erasure-only baseline across all fading conditions.
\end{abstract}

\glsresetall


\IEEEpeerreviewmaketitle

\section{Introduction}
\Ac{pls} exploits the randomness of wireless channels to provide information-theoretic confidentiality, offering security guarantees that hold irrespective of the adversary's computational capability \cite{hamamreh2018classifications}. Beyond passive eavesdropping, \ac{pls} can counter impersonation and message falsification \cite{Liu2017Physical}, and can incorporate semantic information to strengthen security \cite{mitev2023physical}.

Yet conventional \ac{pls} remains fundamentally passive: eavesdroppers can intercept at will with little risk of detection \cite{chaman2018ghostbuster}, while legitimate users bear the full cost of protection \cite{zhou2011throughput,wang2018survey}. Deception technologies address this asymmetry by feeding false data to the attacker, forcing the eavesdropper to act on misleading information or revealing its presence \cite{wang2018cyber}.

We introduced the \ac{pld} framework in \cite{han2023non} and extended it to nonorthogonal multiplexing \cite{chen2025physical} and \ac{ofdm} \cite{chen2025physicalOFDM}. In all three designs, the transmitter encrypts a fraction of its messages so that an eavesdropper who fails to decode the key obtains falsified---but plausible---information. These studies optimized only the transmitter's strategy, treating the receiver as a fixed decoder. In \cite{han2025semantic}, we adopted semantic distortion as the metric and identified three receiver decryption strategies (Perception, Dropping, and Exclusion), but the transmitter's joint optimization of resource allocation and ciphering rate under adaptive receivers remained open.

In this work, we optimize the transmitter’s resource allocation and ciphering rate by accounting for the decryption strategies of both the eavesdropper and the legitimate receiver. Our contributions are as follows:

\begin{enumerate}
    \item We model the transmitter–receiver interaction as a Stackelberg game and derive closed-form strategy switching surfaces that partition the parameter space into regimes where each decryption strategy dominates. We establish an exclusion dominance condition and characterize the sub-regime structure through a triple-point analysis.
    \item We prove that the robust operating point, at the peak of the worst-case distortion envelope, is a Stackelberg equilibrium. The transmitter achieves a provable security guarantee without needing to predict the eavesdropper’s exact strategy. We further show that iterative best-response dynamics oscillate with period two and yield strictly lower security than the equilibrium.
    \item We derive closed-form solutions for the optimal ciphering rate and resource allocation across all nine strategy-pair cases, and propose an algorithm combining the equilibrium with resource optimization.
    \item We evaluate the robust design under Nakagami-$m$ fading with static and channel-adaptive transmitter strategies, and benchmark against a classical \ac{pls} baseline that relies solely on packet erasures. Numerical results confirm the regime boundaries, quantify the security gain over naive iteration, and demonstrate $12$--$55\%$ higher eavesdropper distortion compared to the erasure-only baseline.
\end{enumerate}

The remainder of this paper is organized as follows. Sec.~\ref{sec:related_works} reviews related work. Sec.~\ref{sec:problem_setup} introduces the \ac{pld} framework, decryption strategies, and the semantic distortion metric. Sec.~\ref{sec:strategy_analysis} analyzes the receiver strategy regimes, derives the switching surfaces, and establishes the Stackelberg equilibrium at the robust operating point. Sec.~\ref{sec:strategy_optimization} formulates the resource allocation problem and presents the optimization algorithm. Sec.~\ref{sec:numerical} provides numerical results including a fading channel evaluation with baseline comparison, and Sec.~\ref{sec:conclusion} concludes the paper.

\section{Related Works}
\label{sec:related_works}
\emph{Shannon} \cite{shannon1949communication} showed that perfect secrecy requires shared keys at least as long as the message. \emph{Wyner}’s wiretap channel \cite{Wyner1975wiretap} circumvented this requirement: when the eavesdropping channel is degraded relative to the legitimate channel, positive secrecy capacity exists. \emph{Csiszár} and \emph{Körner} \cite{csiszar2003broadcast} generalized this to non-degraded broadcast channels.
Secrecy capacity has since been characterized for Gaussian \cite{Leung1978Gaussian}, fading \cite{barros2006secrecy}, \ac{mimo} \cite{hero2004secure}, \ac{simo} \cite{parada2005secrecy}, \ac{miso} \cite{li2007secret}, broadcast \cite{liang2008secure}, multiple-access \cite{liu2006discrete}, interference \cite{liang2009capacity}, and relay \cite{oohama2007capacity} channels, under full \cite{bloch2008wireless}, imperfect \cite{li2013spatially}, or absent \cite{khisti2010secure} eavesdropper \ac{csi}. Practical \ac{pls} techniques (power allocation, precoding, antenna selection, artificial noise) are surveyed in \cite{Wang2019survey}.

For the \ac{fbl} regime relevant to short-packet communications, \emph{Polyanskiy} et al.\ \cite{polyanskiy2010channel} derived tight bounds on the decoding error probability, now standard for \ac{urllc} analysis \cite{liu2023energy,li2023robust,soleymani2025optimization}. \emph{Yang} et al.\ \cite{yang2019wiretap} tightened the second-order coding rate bounds for wiretap channels. On the \ac{pls} side, \cite{wang2022achieving} jointly optimized precoding and artificial noise under a covertness constraint, \cite{Oh2023joint} studied the reliability--security trade-off via resource allocation, and \cite{zhu2023trade} introduced leakage-failure probability as a joint reliability--security metric for short packets.

Deception as a defensive tool dates back to \emph{Cheswick}’s \cite{cheswick1992evening} and \emph{Stoll}’s \cite{stoll2024cuckoo} honeypot experiments, and has since been deployed across network, system, and application layers; see \cite{fraunholz2018demystifying,HKB2018deception,PCZ2019game} for surveys. At the physical layer, deception remains less explored. \emph{He} et al.\ \cite{He2023Proactive} used multi-antenna transmission to feed fake but plausible messages to the eavesdropper, \emph{Qi} et al.\ \cite{qi2024adversarial} designed \ac{gan}-based adversarial waveforms that set defense traps, and \emph{Ye} et al.\ \cite{Ye2025Anti} embedded camouflage features into transmitted signals to mislead eavesdropper alignment.

Game theory has a long history in \ac{pls}. Stackelberg formulations, where the transmitter commits first and the adversary best-responds, have been used to optimize power allocation \cite{li2011stackelberg}, beamforming \cite{zhu2016secure}, and friendly jamming \cite{garnaev2016friendly}. Zero-sum and Bayesian games address simultaneous moves or incomplete \ac{csi} \cite{altman2009survey}. These works share a common structure: the adversary’s action space is continuous (power level, beamforming vector), and the equilibrium is found via standard convex optimization or water-filling. In contrast, the \ac{pld} setting yields a \emph{discrete} adversary action space (three decryption strategies) that partitions the parameter space into distinct regimes with sharp transitions, requiring a qualitatively different analysis.

\section{Problem Setup}
\label{sec:problem_setup}

\subsection{System Model}
\label{sec:system_model}
The \ac{pld} framework is shown in Fig.~\ref{fig:PLD_transmitter}. The transmitter \emph{Alice} sends encrypted messages to the legitimate user \emph{Bob} over the legitimate channel $h\bob$, while an eavesdropper \emph{Eve} listens to \emph{Alice} over the eavesdropping channel $h\eve$. With proper beamforming, \emph{Alice} ensures that $h_{Bob}$ remains statistically better than $h_{Eve}$, enabling \ac{pls}.

\begin{figure}[!htpb]
    \centering
    \includegraphics[width=\linewidth]{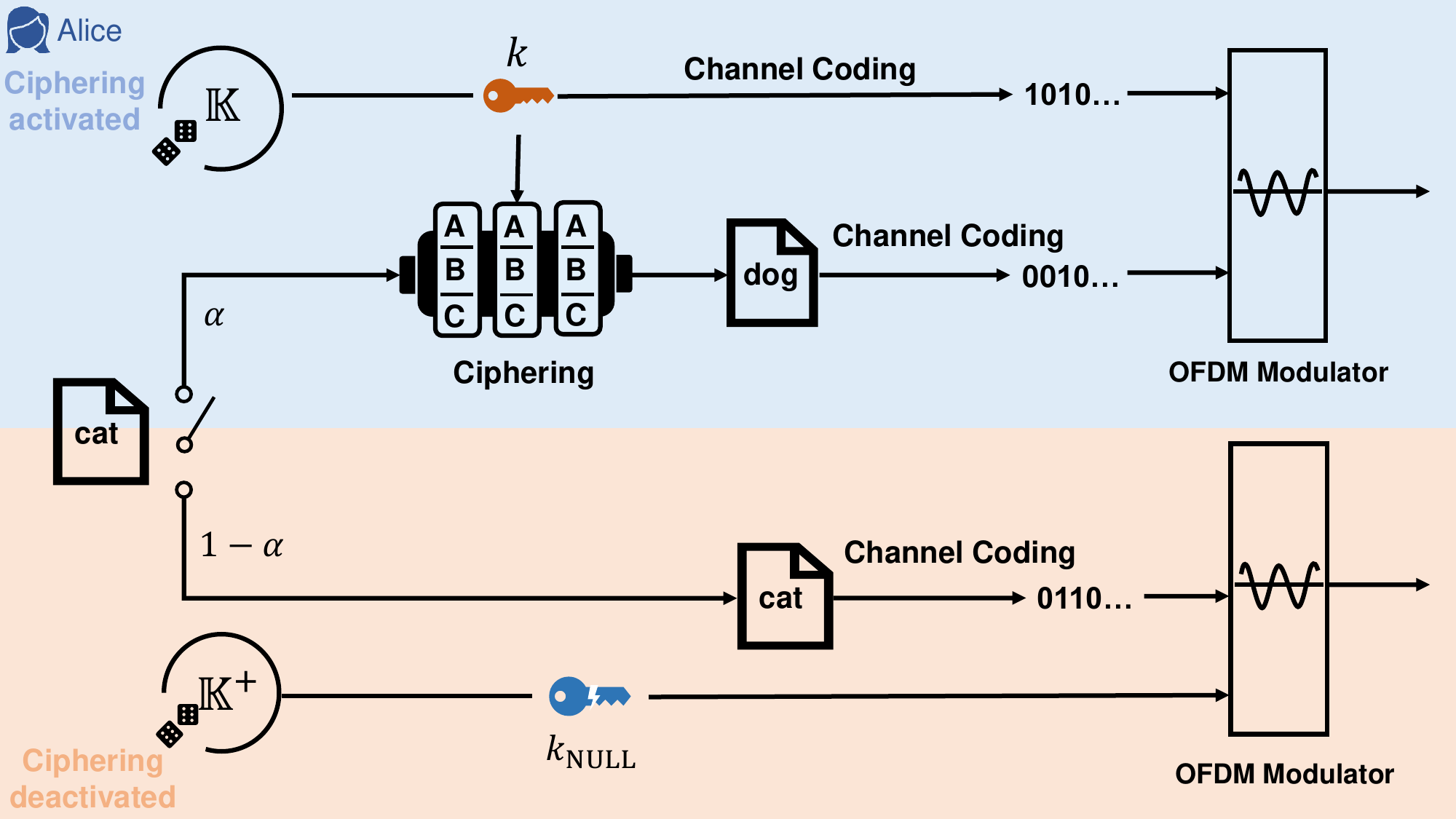}
    \caption{System model of \ac{pld} at the transmitter side}
    \label{fig:PLD_transmitter}
\end{figure}

At the transmitter, \emph{Alice} activates deceptive ciphering with probability $\alpha$. With a random key $k \in \mathbb{K}$, the plaintext $p \in \mathbb{S}$ is encrypted into ciphertext $m \in \mathbb{S}$.  We define the encryption process as:
\begin{equation}
    m=f_k(p) \in \mathbb{S}, \quad \forall (p,k)\in (\mathbb{S}\times\mathbb{K}).
\end{equation}

Conversely, given $k$, the plaintext is recovered as:
\begin{equation}
    p=f^{-1}_k(m).
\end{equation}

When the cipherer is deactivated, \emph{Alice} transmits the plaintext without encryption, i.e., $m=p$. To prevent \emph{Eve} from inferring the cipherer status, \emph{Alice} simultaneously transmits a \emph{dummy sequence} (hereafter called a \emph{litter sequence}) $k\subscript{NULL}$ that occupies the key channel, making the transmission indistinguishable from the encrypted case. 

{The codebook $\mathbb{S}$ must be designed such that both the plaintext and the ciphertext are valid elements within it. If \emph{Eve} fails to decode the key, she interprets the ciphertext as the original message. Since both are valid codebook entries with plausible semantic meaning, \emph{Eve} has no reason to question authenticity. As a concrete example, for $\mathcal{S}=2$ the codebook contains two messages and a single key bit flips between them; for larger $\mathcal{S}$, any permutation cipher on $\mathbb{S}$ indexed by $k$ satisfies the requirement.} 

We consider a worst-case scenario in which the eavesdropper has full knowledge of the \ac{pld} mechanisms, and the modulation and channel coding schemes are publicly known and shared among \emph{Alice}, \emph{Bob}, and \emph{Eve}. This models threats such as a breached database or malicious insider.

At the receiver side, \revise{the receiver}{\emph{Bob} and \emph{Eve}} decode both the ciphertext and the key simultaneously, as illustrated in Fig.~\ref{fig:PLD_receiver}. If ciphertext decoding fails, the receiver discards the message. Otherwise, the receiver’s action depends on the key decoding outcome. When the key is also correctly decoded, the corresponding plaintext $\hat{p}$ can be recovered from $\hat{m}$ and $\hat{k}$. Otherwise, if the key decoding fails, the receiver may adopt one of three possible decryption strategies: 
\begin{enumerate}
    \item \emph{Perception}: The receiver assumes no valid key was transmitted and interprets the ciphertext as plaintext, i.e., $\hat{m}=\hat{p}$. This is the default for deterministic decryption.
    \item \emph{Dropping}: The receiver cannot determine whether the failure of decoding is due to poor channel quality preventing the decoding of a valid key, or because \emph{Alice} transmitted $k\subscript{NULL}$. Thus, the receiver simply drops the received message $\hat{m}$, treating it as $\hat{p}=s\subscript{NULL}$.
    \item \emph{Exclusion}: The receiver assumes that a valid key was transmitted but not correctly decoded, which means that the ciphertext is not equal to the plaintext. Thus, the receiver should select another codeword $\hat{p}\neq \hat{m}$ for $\hat{p} \in \mathbb{S}$. We assume that $\Pr(p)=\frac{1}{\mathcal{S}}$ for all $p \in \mathbb{S}$, where $\mathcal{S}$ is the codebook cardinality. In this case, the optimal strategy to select $\hat{p}$ shall follow the principle of maximum likelihood:
    \begin{equation}
        f_{\hat{k}}^{-1}(\hat{m})=\arg \underset{p\in \mathbb{S}\backslash\{\hat{m}\}}{\max} \Pr(p).
    \end{equation}
\end{enumerate}

\begin{figure}[!htpb]
    \centering
    \includegraphics[width=\linewidth]{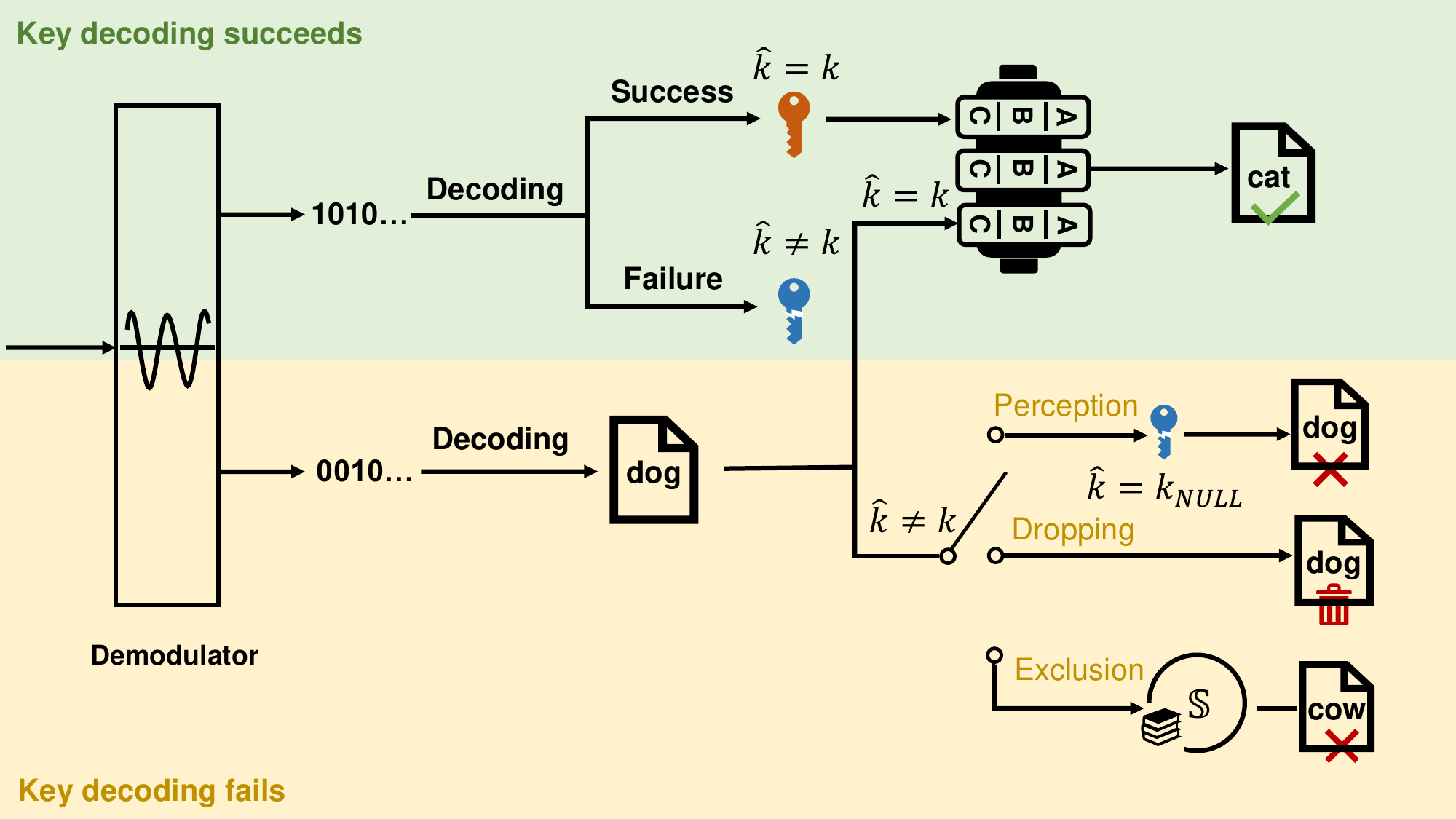}
    \caption{Decryptor model at the receiver side}
    \label{fig:PLD_receiver}
\end{figure}

{The three strategies apply to both \emph{Bob} and \emph{Eve}.} Due to the superior conditions of the legitimate channel, \emph{Bob} is highly likely to correctly decode both the ciphertext and the key, thereby recovering the intended plaintext. In contrast, \emph{Eve} is less likely to decode the key successfully and is consequently more inclined to employ one of the three decryption strategies. With adequate redundancy and power, the likelihood of mistaking one valid codeword for another is negligible compared to the probability of a decoding failure. Therefore, in the following analysis, we only consider the case of \emph{Erasure} for decoding error, where bit errors exceed the correction capability but remain within the detection capability.

Under Perception or Exclusion without a recovered key, the receiver obtains incorrect information. Therefore, we can design codebooks so that acting on the falsely decrypted information incurs a semantic penalty on \emph{Eve}, quantified by the semantic distortion metric defined below.

\subsection{Performance Metrics}
The \ac{pld} model admits a semantic communication interpretation. The plaintext $p\in\mathbb{S}$ represents a meaning and the ciphertext $m \in \mathbb{S}$ a message. \emph{Alice} encodes meaning $p\in \mathbb{S}$ into message $m \in \mathbb{S}$ with key $k \in \mathbb{K}$. The encryptor is thus a semantic encoder~\cite{shao2024theory}, denoted $u_k$ (Fig.~\ref{fig:semantic_model_pld}).

When the deceptive ciphering is deactivated, \emph{Alice} sends the unciphered plaintext $p \in \mathbb{S}$ together with a litter sequence $k\subscript{NULL}$. We define $\mathbb{K}^+=\mathbb{K}\cup\{k\subscript{NULL}\}$ so that every message is encoded by $u_k$ with $k\in\mathbb{K}^+$. In particular, we have $u_{k\subscript{NULL}}(p)=p$ for all $p\in\mathbb{S}$.

At the receiver, the semantic decoder $v_{\hat{k}}$ uses the received key $\hat{k}\in \mathbb{K}^+$ to decode the received ciphertext $\hat{m} \in
\mathbb{S}^+$ into the meaning $\hat{p}$, where $\mathbb{S}^+=\mathbb{S}\cup\{s\subscript{NULL}\}$ and $s\subscript{NULL}$ denotes the decoding error flag.

\begin{figure}[!htpb]
    \centering
    \includegraphics[width=\linewidth]{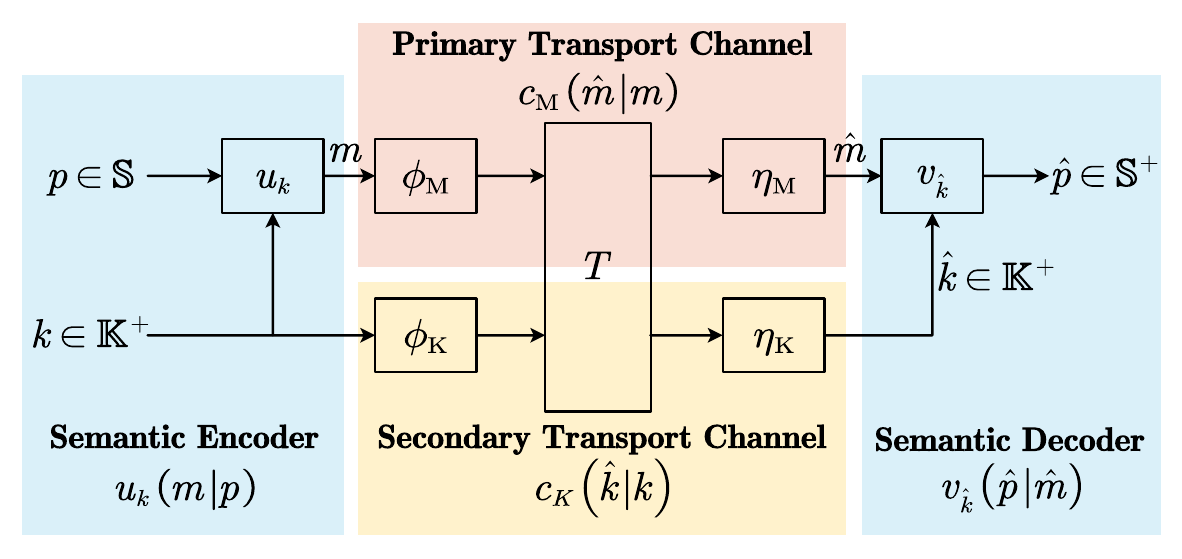}
    \caption{Dual-channel model of \ac{pld}}
    \label{fig:semantic_model_pld}
\end{figure}

We define the ciphertext channel as the \emph{primary transport channel} and the key channel as the \emph{secondary transport channel}. In the primary channel, $m$ passes through encoder $\Phi\subscript{M}$, physical channel $T$, and decoder $\eta\subscript{M}$ to yield $\hat{m}$.
 As noted above, \emph{confusion}—decoding one codeword as another—is negligible with appropriate coding. Therefore, the \emph{primary transport channel} can be modeled as an erasure channel, as illustrated in Fig. \ref{fig:primary_transport_channel}.

\begin{figure}
    \centering
    \includegraphics[width=0.8\linewidth,clip,trim=0 10 0 10]{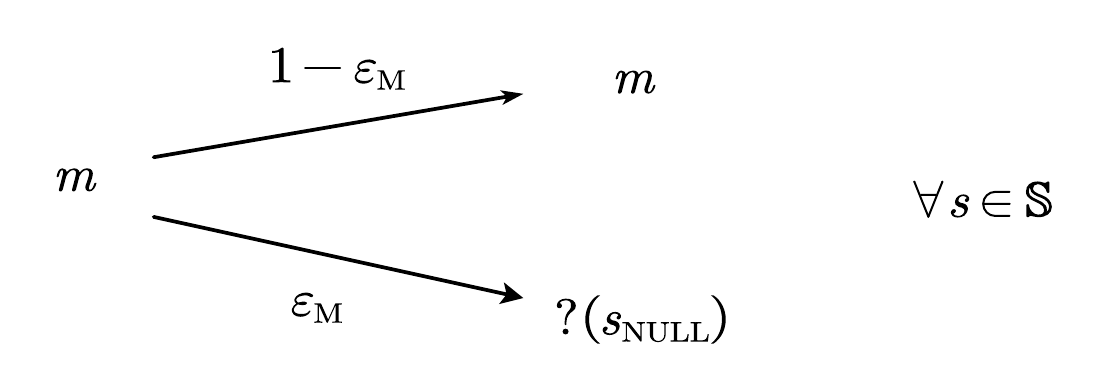}
    \caption{Model of the primary transport channel}
    \label{fig:primary_transport_channel}
\end{figure}

The conditional \ac{pdf} of the primary transport channel is:
\begin{equation}
    \begin{aligned}
 c\subscript{M}(\hat{m}|m) &=   \begin{cases}
 1-\varepsilon\subscript{M},    & \hat{m}=m \\
 \varepsilon\subscript{M}, & \hat{m}=s\subscript{NULL} \\
 0, & \text{otherwise}
\end{cases} \\
&= (1-\varepsilon\subscript{M})\delta(\hat{m}-m)+\varepsilon\subscript{M}\delta(\hat{m}-s\subscript{NULL}).
   \end{aligned}
\end{equation}

Since litter sequences are random, the secondary transport channel has a Z-shape (Fig.~\ref{fig:secondary_transport_channel}). The conditional \ac{pdf} of the \emph{secondary transport channel} can be formulated as:
\begin{equation}
    \begin{aligned}
    c\subscript{K}(\hat{k}|k)=&
        \begin{cases}
            1-\varepsilon\subscript{K}, & \hat{k}=k, k \in \mathbb{K} \\
            \varepsilon\subscript{K}, & \hat{k}=k\subscript{NULL}\neq k, k \in \mathbb{K} \\
            1, &\hat{k}=k=k\subscript{NULL} \\
            0, & \text{otherwise}
        \end{cases} \\
        =&\left[(1-\varepsilon\subscript{K})\delta(\hat{k}-k)+\varepsilon\subscript{K}\delta(\hat{k}-k\subscript{NULL})\right]\cdot\\
        &[1-\delta(k-k\subscript{NULL})]+\delta(\hat{k}-k\subscript{NULL}).
    \end{aligned}
\end{equation}

\begin{figure}
    \centering
    \includegraphics[width=0.8\linewidth,clip,trim=0 10 0 10]{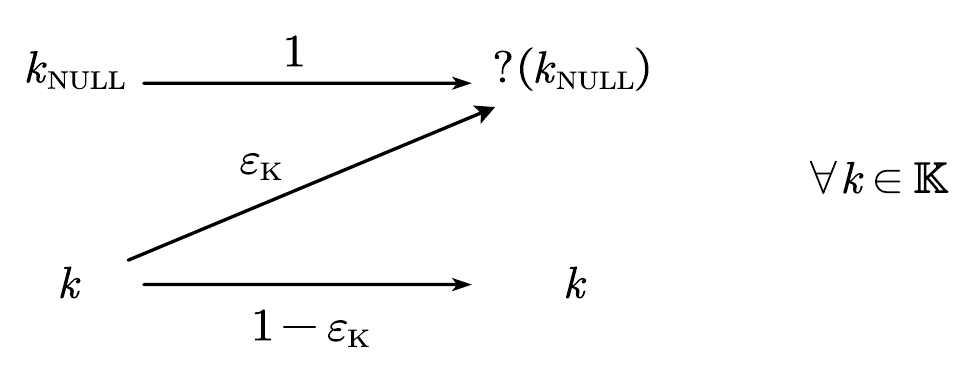}
    \caption{Model of the secondary transport channel}
    \label{fig:secondary_transport_channel}
\end{figure}

We define the semantic distortion between meanings $p_1, p_2 \in \mathbb{S}^+$ as $d(p_1,p_2)$. Thus, $d(p,p)\equiv 0, \forall p\in \mathbb{S}^+$. 

From the two transport channels and the semantic codec, the end-to-end conditional \ac{pdf} is:
\begin{equation}
    \psi(\hat{p}|p)=\sum_{p,\hat{p},k,\hat{k}}\Pr(k)u_k(m|p)c\subscript{M}(\hat{m}|m)c\subscript{K}(\hat{k}|k)v_{\hat{k}}(\hat{p}|\hat{m}),
\end{equation}
and the average semantic distortion is:
\begin{equation}
    D=\sum_{p,\hat{p}}\Pr(p)\psi(\hat{p}|p)d(p,\hat{p}).
    \label{eq:distortion_with_key}
\end{equation}

Given $k$, the encoder and decoder are deterministic: $m = f_k(p)$, $\hat{p} = f_k^{-1}(\hat{m})$. Their conditional distributions are:
\begin{equation}
 u_k(m|p) =
    \begin{cases}
       1, & m=f_k(p), \\
       0, & \text{otherwise},
    \end{cases}
\end{equation}
\begin{equation}
 v_{\hat{k}}(\hat{p}|\hat{m}) =
    \begin{cases}
       1, & \hat{p}=f^{-1}_{\hat{k}}(\hat{m}), \\
       0, & \text{otherwise},
    \end{cases}
\end{equation}
Introducing $\alpha\in[0,1]$ as the encryption probability and accounting for erasures, \eqref{eq:distortion_with_key} becomes:
\begin{equation}
\begin{split}
    D&=\varepsilon_{\mathrm{M}}d(p,s\subscript{NULL})\\
    &+(1-\varepsilon_{\mathrm{M}})\varepsilon\subscript{K}\sum_p \Pr(p)\sum_{k \in \mathbb{K}}\Pr(k)d(p,f_k(p)).
    \label{eq:Distortion_deterministic}
\end{split}
\end{equation}

\revise{The distortion $d(p,s\subscript{NULL})$ and $d(p,f_k(p))$ can be determined from the codebook of plaintext $p$ and ciphertext $m$, and thus can be considered as constant.}{Once the codebook is fixed, $d(p,s\subscript{NULL})$ and $d(p,f_k(p))$ are constants. Defining $D\subscript{loss}\triangleq d(p,s\subscript{NULL})$ (loss distortion) and $D\subscript{conf}\triangleq d(p,f_k(p))$ (confusion distortion), \eqref{eq:Distortion_deterministic} simplifies to:}
\begin{equation}
\label{eq:D_i_resource_allocation}
D_i=\varepsilon_{i,\mathrm{M}}D\subscript{loss}+\alpha (1-\varepsilon_{i,\mathrm{M}})\varepsilon_{i,\mathrm{K}}D\subscript{conf},\quad i\in\{\text{Bob},\text{Eve}\}.
\end{equation}

\subsection{Decryptor Model}
Deception occurs when the receiver fails to decode the key, contributing $D\subscript{conf}$. When $\hat{k}=k\subscript{NULL}$, the receiver chooses among Perception, Dropping, and Exclusion (Fig.~\ref{fig:PLD_receiver}). We assume $\Pr(p)=\frac{1}{\mathcal{S}}$, where $\mathcal{S}$ denotes the codebook cardinality, with selection probabilities $(\beta_1, \beta_2, \beta_3)$. The stochastic semantic decoder is:
\begin{equation}
\Tilde{v}_{\hat{k}}(\hat{p}|\hat{m})=
    \begin{cases}
        1, & \hat{p}=f_{\hat{k}}^{-1}(\hat{m}),\hat{k}\neq k\subscript{NULL} \\
        \beta_1, & \hat{p}=\hat{m},\hat{k}=k\subscript{NULL} \\
        \beta_2, & \hat{p}=m\subscript{NULL},\hat{k}=k\subscript{NULL} \\
        \frac{\beta_3}{\mathcal{S}-1}, & \hat{p}\in \mathbb{S}\backslash \{\hat{m}\},\hat{k}=k\subscript{NULL} \\
        0, &\text{otherwise}
    \end{cases}.
\end{equation}

The distortion under the opportunistic decryptor is:
\begin{equation}
\label{eq:Tilde_D_i}
    \Tilde{D}_i = \varepsilon_{i,\mathrm{M}} D\subscript{loss} + (1-\varepsilon_{i,\mathrm{M}}) \sum_{k=1}^{3} \beta_k \Delta_k,
\end{equation}
where the three branch distortions are
\begin{equation}
\label{eq:Delta_k}
\begin{aligned}
    \Delta_1 &= \varepsilon_{i,\mathrm{K}} \alpha D\subscript{conf}, \\
    \Delta_2 &= \left[\varepsilon_{i,\mathrm{K}} \alpha + (1-\alpha)\right] D\subscript{loss}, \\
    \Delta_3 &= \left[\varepsilon_{i,\mathrm{K}} \frac{\alpha(\mathcal{S}-2)}{\mathcal{S}-1} + (1-\alpha)\right] D\subscript{conf}.
\end{aligned}
\end{equation}

Given channel conditions and resource allocation, $\varepsilon_{i,\mathrm{M}}$ and $\varepsilon_{i,\mathrm{K}}$ are determined for $i \in \{\mathrm{Bob, Eve}\}$. Each receiver selects $(\beta_1^i,\beta_2^i,\beta_3^i)$ to minimize $\Tilde{D}_i$:
 \begin{mini!}
    	{\beta_1^i,\beta_2^i,\beta_3^i}{\Tilde{D}_{i}=\Tilde{D}_{i}(\varepsilon_{i,\mathrm{M}},\varepsilon_{i,\mathrm{K}}) \label{obj:Tilde_D_i}} {\label{Problem_D_i}}{}
        \addConstraint{\beta_1^i+\beta_2^i+\beta_3^i=1}.
    \end{mini!}
    
    \revise{This optimization problem can be simply solved by comparing the distortions of the three options. Thus, the optimal $\beta_k$ must be associated with the minimal distortion.}{ This linear program is solved by setting $\beta_k=0$ for all non-minimal $\Delta_k$ and distributing the remaining weight among minimal ones.} Thus, the optimum should satisfy:
 \begin{equation} 
   \begin{aligned}
        \beta_k&=0, \quad \forall \Delta_k\neq\Delta_{\min} \\
        \sum_{k:\Delta_k=\Delta_{\min}}\beta_k&=1, 
    \end{aligned}
\end{equation}
where $\Delta_{\min}=\min\{\Delta_1,\Delta_2,\Delta_3\}$.

\emph{Remark on the uniform prior.} The assumption $\Pr(p) = 1/\mathcal{S}$ is adopted for analytical tractability. Natural semantic sources are typically skewed (e.g., Zipf-distributed), in which case the maximum-likelihood decoder biases its guess toward high-probability messages and $D\subscript{conf}$ becomes a source-weighted average over the confusion set. The analysis structure is preserved under any $\Pr(p)$: each $\Delta_k(\alpha)$ stays linear in $\alpha$, so the concave piecewise-linear lower envelope and the Stackelberg equilibrium at its peak persist. Only the numerical values of the switching boundaries shift, and the closed-form expressions for $\alpha^{\dagger}$ would need to be re-derived accordingly. A full treatment under non-uniform priors is left for future work.

\section{Receiver Strategy Regime Analysis}
\label{sec:strategy_analysis}

The three decryption strategies from Sec.~\ref{sec:problem_setup} are each optimal under different conditions. Since Alice designs the transmission parameters $(n\subscript{M}, n\subscript{K}, \alpha)$ before the receivers act, she must account for worst-case Eve behavior. We derive closed-form switching surfaces partitioning the parameter space into strategy regimes, then establish Alice's robust operating point and its game-theoretic optimality.

\subsection{Strategy Switching Surfaces}
\label{subsec:switching_surfaces}

Each receiver selects the strategy minimizing $\Delta_k$ in \eqref{eq:Tilde_D_i}. The \emph{strategy switching surfaces} are the loci where two strategies yield equal distortion.

For notational convenience, we introduce the distortion ratio $\rho \triangleq D\subscript{loss} / D\subscript{conf} \in (0,1)$.

\begin{lemma}[Strategy Switching Surfaces]
\label{lem:switching_surfaces}
The switching surfaces $\mathcal{B}_{jk}$ are defined by $\Delta_j = \Delta_k$ for $j,k \in \{1,2,3\}$, $j \neq k$:
\begin{align}
    \mathcal{B}_{12}&: \quad \alpha_{12}^* = \frac{\rho}{\varepsilon\subscript{K}(1-\rho) + \rho}, \label{eq:B12} \\
    \mathcal{B}_{13}&: \quad \alpha_{13}^* = \frac{\mathcal{S}-1}{\varepsilon\subscript{K} + \mathcal{S}-1}, \label{eq:B13} \\
    \mathcal{B}_{23}&: \quad \alpha_{23}^* = \frac{1-\rho}{\varepsilon\subscript{K}\left(\rho - \frac{\mathcal{S}-2}{\mathcal{S}-1}\right) + 1-\rho}, \label{eq:B23}
\end{align}
where $\mathcal{B}_{23}$ exists only when $\rho > \frac{\mathcal{S}-2}{\mathcal{S}-1}$.
\end{lemma}

\noindent\emph{Proof:} See Appendix~\ref{app:proof_switching_surfaces}.\footnote{Due to space constraints, all proofs are provided in the accompanying appendix document.}

All three thresholds decrease in $\varepsilon\subscript{K}$: a worse key channel makes Perception less attractive, so the receiver switches away at a lower $\alpha$. Since $\varepsilon\subscript{Bob,K} < \varepsilon\subscript{Eve,K}$ under the legitimate channel advantage, we have:
\begin{equation}
\label{eq:threshold_ordering}
    \alpha_{jk}^*(\varepsilon\subscript{Eve,K}) < \alpha_{jk}^*(\varepsilon\subscript{Bob,K}), \quad \forall j,k.
\end{equation}
That is, \emph{Eve switches strategy at a lower $\alpha$ than Bob.}

\subsection{Exclusion Dominance Condition}
\label{subsec:exclusion_dominance}

In~\cite{han2025semantic}, we observed that the Exclusion strategy vanishes for large codebooks but did not quantify the condition. The following theorem provides an exact threshold.

\begin{theorem}[Exclusion Dominance]
\label{thm:exclusion_dominance}
The Exclusion strategy ($\beta_3 = 1$) yields strictly higher distortion than Dropping ($\beta_2 = 1$) for all $(\alpha, \varepsilon\subscript{K}) \in (0,1)^2$ if and only if
\begin{equation}
\label{eq:exclusion_dominance}
    \mathcal{S} > \mathcal{S}^*(\rho) \triangleq \frac{2-\rho}{1-\rho}.
\end{equation}
\end{theorem}

\noindent\emph{Proof:} See Appendix~\ref{app:proof_exclusion_dominance}.

Theorem~\ref{thm:exclusion_dominance} partitions the $(\mathcal{S}, \rho)$ parameter space into two fundamental regimes:
\begin{itemize}
    \item \textbf{Region~A} ($\mathcal{S} > \mathcal{S}^*(\rho)$): Only Perception and Dropping are relevant; the analysis reduces to two strategies.
    \item \textbf{Region~B} ($\mathcal{S} \leq \mathcal{S}^*(\rho)$): All three strategies can be optimal, depending on $\alpha$ and $\varepsilon\subscript{K}$.
\end{itemize}

Table~\ref{tab:critical_S} lists $\mathcal{S}^*(\rho)$; for $\rho = 0.1$, $\mathcal{S}^*= 2.11$, so $\mathcal{S} \geq 3$ falls in Region~A while $\mathcal{S} = 2$ lies in Region~B (Fig.~\ref{fig:phase_diagram}).

\begin{table}[!htpb]
    \centering
    \caption{Critical codebook size $\mathcal{S}^*(\rho)$ for different $\rho$}
    \label{tab:critical_S}
    \begin{tabular}{c|cccccccc}
    \toprule[2px]
    $\rho$ & 0.1 & 0.2 & 0.3 & 0.5 & 0.7 & 0.9 & 0.95 \\
    \midrule[1px]
    $\mathcal{S}^*$ & 2.11 & 2.25 & 2.43 & 3 & 4.33 & 11 & 21 \\
    \bottomrule[2px]
    \end{tabular}
\end{table}

\begin{figure}[!htpb]
    \centering
    \includegraphics[width=.9\linewidth]{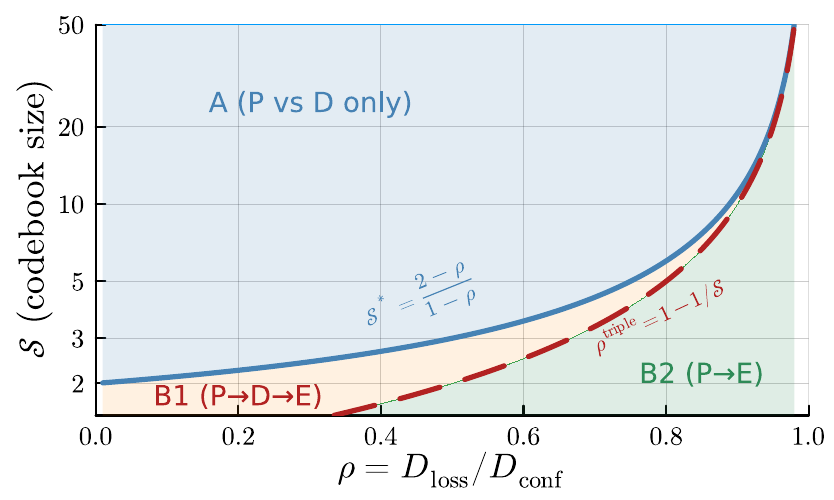}
    \caption{Strategy regime phase diagram in the $(\mathcal{S}, \rho)$ plane.}
    \label{fig:phase_diagram}
\end{figure}

\subsection{Strategy Regime Characterization within Region~B}
\label{subsec:region_B}

Within Region~B, all three strategies can be optimal; dominance depends on the threshold ordering.

\begin{lemma}[Triple Point and Sub-regime Classification]
\label{lem:triple_point}
In Region~B, the three switching surfaces $\mathcal{B}_{12}$, $\mathcal{B}_{13}$, $\mathcal{B}_{23}$ coincide at the \emph{triple point}
\begin{equation}
\label{eq:triple_point}
    \rho\superscript{triple} = 1 - \frac{1}{\mathcal{S}}.
\end{equation}
The sub-regime classification is:
\begin{itemize}
    \item Sub-regime~B1 ($\rho < \rho\superscript{triple}$): $\alpha_{12}^* < \alpha_{13}^* < \alpha_{23}^*$. As $\alpha$ increases, the optimal strategy transitions $P \to D \to E$.
    \item Sub-regime~B2 ($\rho > \rho\superscript{triple}$): $\alpha_{23}^* < \alpha_{13}^* < \alpha_{12}^*$. The Dropping region is empty; only $P \to E$ transition occurs.
\end{itemize}
\end{lemma}

\noindent\emph{Proof:} See Appendix~\ref{app:proof_triple_point}.

The two codebook-size thresholds in Theorem~\ref{thm:exclusion_dominance} and Lemma~\ref{lem:triple_point} describe different geometric events. $\mathcal{S}^*(\rho)$ is a worst-case bound obtained by maximizing the dominance inequality over the full domain $(\alpha, \varepsilon\subscript{K}) \in (0,1)^2$, and gives the smallest codebook size for which Exclusion is nowhere optimal. The triple point $\rho\superscript{triple} = 1 - 1/\mathcal{S}$, by contrast, identifies a single location in the parameter space where the three switching surfaces meet. The triple point always lies strictly inside Region~B (below $\mathcal{S}^*$): crossing it re-orders the sub-regime structure from B1 to B2 but does not eliminate Exclusion, whereas crossing $\mathcal{S}^*$ removes Exclusion from the entire parameter space.

\subsection{Viable Strategy Pairs under Channel Advantage}
\label{subsec:viable_pairs}

Since $\varepsilon\subscript{Bob,K} < \varepsilon\subscript{Eve,K}$ and the switching thresholds are monotone in $\varepsilon\subscript{K}$, not all (Bob, Eve) strategy pairs can co-occur at a given $\alpha$.

\begin{lemma}[Impossible Strategy Pairs]
\label{lem:impossible_pairs}
Under the legitimate channel advantage $\varepsilon\subscript{Bob,K} < \varepsilon\subscript{Eve,K}$, any strategy pair where Bob has switched to a more aggressive strategy than Eve is impossible. Specifically:
\begin{enumerate}[label=(\roman*)]
    \item $(Bob\!=\!D, Eve\!=\!P)$ is impossible in all regimes, since $\alpha_{12}^*(\varepsilon\subscript{Eve,K}) < \alpha_{12}^*(\varepsilon\subscript{Bob,K})$.
    \item $(Bob\!=\!E, Eve\!=\!P)$ is impossible in Region~B, since $\alpha_{13}^*(\varepsilon\subscript{Eve,K}) < \alpha_{13}^*(\varepsilon\subscript{Bob,K})$.
    \item $(Bob\!=\!E, Eve\!=\!D)$ is impossible in Region~B, since $\alpha_{23}^*(\varepsilon\subscript{Eve,K}) < \alpha_{23}^*(\varepsilon\subscript{Bob,K})$.
\end{enumerate}
\end{lemma}

\noindent\emph{Proof:} See Appendix~\ref{app:proof_impossible_pairs}.

In Region~A, this reduces the viable strategy pairs from four to three: $\{(P,P), (P,D), (D,D)\}$. In Region~B (sub-regime B1), six of the nine pairs are viable: $\{(P,P), (P,D), (P,E), (D,D), (D,E), (E,E)\}$.

\subsection{Robust Transmitter Design}
\label{subsec:robust_design}

Alice selects the parameters to maximize the \emph{worst-case} distortion experienced by Eve over all strategy choices.

\begin{definition}[Worst-Case Eve Distortion]
\label{def:wc_distortion}
For given error probabilities $(\varepsilon\subscript{Eve,M}, \varepsilon\subscript{Eve,K})$ and ciphering rate $\alpha$, the worst-case Eve distortion is
\begin{equation}
\label{eq:wc_distortion}
    \Tilde{D}\eve\superscript{wc}(\alpha) \triangleq \varepsilon\subscript{Eve,M} D\subscript{loss} + (1-\varepsilon\subscript{Eve,M}) \min_{k \in \{1,2,3\}} \Delta_k(\alpha, \varepsilon\subscript{Eve,K}).
\end{equation}
\end{definition}

The $\min$ ensures that any Eve deviation yields even higher distortion.
Alice's robust design problem is:
\begin{maxi!}
    {n\subscript{M}, n\subscript{K}, \alpha}{\Tilde{D}\eve\superscript{wc}(\alpha)}
    {\label{prob:robust_design}}{}
    \addConstraint{\Tilde{D}\bob^* \leq \Tilde{D}\superscript{th}\bob \label{con:bob_threshold}}
    \addConstraint{n\subscript{M}, n\subscript{K} \in \mathbb{Z}^+, \; \alpha \in [0,1] \label{con:feasibility}},
\end{maxi!}
where $\Tilde{D}\bob^* = \min_k \Delta_k(\alpha, \varepsilon\subscript{Bob,K})$ is Bob's distortion under his own optimal (cooperative) decryption.

Among the three branches, $\Delta_1(\alpha)$ increases strictly in $\alpha$ with slope $\varepsilon\subscript{Eve,K} D\subscript{conf}$, while $\Delta_2(\alpha)$ and $\Delta_3(\alpha)$ decrease strictly, with slopes $-(1-\varepsilon\subscript{Eve,K})D\subscript{loss}$ and $-[1 - \varepsilon\subscript{Eve,K}(\mathcal{S}-2)/(\mathcal{S}-1)]D\subscript{conf}$ respectively. Their lower envelope $\min_k \Delta_k(\alpha)$ is a concave piecewise-linear tent, and its maximum lies at the active switching boundary. This monotonicity structure underpins the closed-form solution below.

\begin{theorem}[Robust Operating Point]
\label{thm:robust_alpha}
The worst-case-optimal ciphering rate for fixed $(\varepsilon\subscript{Eve,M}, \varepsilon\subscript{Eve,K})$ is
\begin{equation}
\label{eq:robust_alpha}
    \alpha^{\dagger} = \begin{cases}
    \dfrac{\rho}{\varepsilon\subscript{Eve,K}(1-\rho) + \rho}, & \text{if } \rho \leq 1 - \dfrac{1}{\mathcal{S}}, \\[10pt]
    \dfrac{\mathcal{S}-1}{\varepsilon\subscript{Eve,K} + \mathcal{S}-1}, & \text{if } \rho > 1 - \dfrac{1}{\mathcal{S}},
    \end{cases}
\end{equation}
and the corresponding worst-case distortion (excluding the $\varepsilon\subscript{Eve,M}$ term) is
\begin{equation}
\label{eq:robust_distortion}
    \Delta^{\dagger} = \begin{cases}
    \dfrac{\varepsilon\subscript{Eve,K} \rho D\subscript{conf}}{\varepsilon\subscript{Eve,K}(1-\rho) + \rho}, & \text{if } \rho \leq 1 - \dfrac{1}{\mathcal{S}}, \\[10pt]
    \dfrac{\varepsilon\subscript{Eve,K} (\mathcal{S}-1) D\subscript{conf}}{\varepsilon\subscript{Eve,K} + \mathcal{S}-1}, & \text{if } \rho > 1 - \dfrac{1}{\mathcal{S}}.
    \end{cases}
\end{equation}
Both expressions coincide at $\rho = 1 - 1/\mathcal{S}$.
\end{theorem}

\noindent\emph{Proof:} See Appendix~\ref{app:proof_robust_alpha}.

\begin{corollary}[Security Guarantee]
\label{cor:security_guarantee}
At the robust operating point $\alpha^{\dagger}$, Eve's distortion under any strategy satisfies
\begin{equation}
    \Tilde{D}\eve(\alpha^{\dagger}, \beta_k=1) \geq \Tilde{D}\eve\superscript{wc}(\alpha^{\dagger}), \quad \forall k \in \{1,2,3\},
\end{equation}
where the equality holds with a distortion-minimizing Eve.
\end{corollary}

The robust operating point provides a guaranteed lower bound on Eve's distortion (Fig.~\ref{fig:lower_envelope}).

\subsection{Stackelberg Equilibrium and Best-Response Dynamics}
\label{subsec:stackelberg}

The robust operating point $\alpha^{\dagger}$ is a Stackelberg equilibrium of the game where Alice (leader) commits to $(n\subscript{M}, n\subscript{K}, \alpha)$ and Eve (follower) best-responds.

\begin{theorem}[Stackelberg Equilibrium]
\label{thm:stackelberg}
The strategy profile $(\alpha^{\dagger}, \beta^*\eve, \beta^*\bob)$, where $\alpha^{\dagger}$ is the robust operating point from Theorem~\ref{thm:robust_alpha}, constitutes a Stackelberg equilibrium of the \ac{pld} design game. Specifically:
\begin{enumerate}[label=(\roman*)]
    \item \textbf{Eve's best response:} At $\alpha = \alpha^{\dagger}$, Eve is indifferent between the two strategies adjacent to the active switching boundary, i.e., $\Delta_j(\alpha^{\dagger}) = \Delta_k(\alpha^{\dagger})$ for the binding pair $(j,k)$. Any mixture over them two is a best response.
    \item \textbf{Bob's consistency:} The threshold ordering $\alpha_{12}^*(\varepsilon\subscript{Eve,K}) < \alpha_{12}^*(\varepsilon\subscript{Bob,K})$ from \eqref{eq:threshold_ordering} ensures that Bob remains in Perception at $\alpha^{\dagger}$, since $\alpha^{\dagger} = \alpha_{12}^*(\varepsilon\subscript{Eve,K}) < \alpha_{12}^*(\varepsilon\subscript{Bob,K})$.
    \item \textbf{Alice's optimality:} Alice cannot increase $\Tilde{D}\eve\superscript{wc}$ by unilaterally deviating from $\alpha^{\dagger}$, since $\alpha^{\dagger}$ maximizes the concave lower envelope $\min_k \Delta_k(\alpha)$ by Theorem~\ref{thm:robust_alpha}.
\end{enumerate}
\end{theorem}

\noindent\emph{Proof:} See Appendix~\ref{app:proof_stackelberg}.

\emph{Remark (tie-breaking irrelevance).} Eve's best response at $\alpha^{\dagger}$ is not unique, which normally requires distinguishing a strong Stackelberg equilibrium, where the follower resolves ties in the leader's favor, from a weak one, where the follower resolves ties against the leader. The distinction is void here. Alice's payoff equals $\min_k \Delta_k(\alpha)$, and at $\alpha^{\dagger}$ every pure strategy in the minimizing set, as well as every mixture over them, yields the same value $\Delta^{\dagger}$. Theorem~\ref{thm:stackelberg} therefore holds under any tie-breaking rule.

Alice need not predict Eve's strategy, only that Eve minimizes her distortion. Deviations only increase Eve's distortion (Corollary~\ref{cor:security_guarantee}).

\begin{theorem}[Iterative Best-Response Suboptimality]
\label{thm:best_response}
If Alice and Eve alternate best responses (Alice solving the \ac{lp} relaxation of \eqref{prob:robust_design} for a fixed Eve strategy, Eve updating given Alice's $\alpha$), the iteration does not converge:
\begin{enumerate}[label=(\roman*)]
    \item \textbf{Overshoot:} The \ac{lp} boundary solutions push $\alpha$ to extremes: $\alpha\superscript{o} \gg \alpha^{\dagger}$ when assuming Perception, and $\alpha\superscript{o} = 0 < \alpha^{\dagger}$ when assuming a defensive strategy. This overshoots the switching surface $\mathcal{B}_{12}$ (or $\mathcal{B}_{13}$) in both directions.
    \item \textbf{Periodicity:} The overshoot induces a strategy switch at each iteration. Since $\alpha$ alternates between a high value (above $\alpha^{\dagger}$) and zero (below $\alpha^{\dagger}$), the dynamics are eventually periodic with period~$2$.
    \item \textbf{Suboptimality:} The time-averaged Eve distortion over one oscillation cycle satisfies
    \begin{equation}
    \label{eq:suboptimality}
        \bar{D}\eve = \frac{1}{2}\bigl[\Delta_j(\alpha\superscript{high}) + \Delta_k(\alpha\superscript{low})\bigr] < \Delta^{\dagger} = \Tilde{D}\eve\superscript{wc}(\alpha^{\dagger}),
    \end{equation}
    where $\alpha\superscript{high}$ and $\alpha\superscript{low}$ are the two iterates and $\Delta^{\dagger}$ is the robust distortion from \eqref{eq:robust_distortion}. The strict inequality follows from the concavity of the lower envelope.
\end{enumerate}
\end{theorem}

\noindent\emph{Proof:} See Appendix~\ref{app:proof_best_response}.

The \ac{lp} drives each iterate to a boundary overshooting $\alpha^{\dagger}$, whereas the equilibrium is achieved in closed form via Theorem~\ref{thm:robust_alpha}. Sec.~\ref{subsec:oscillation} quantifies the gap.

\begin{figure*}[!htpb]
    \centering
    \begin{subfigure}[b]{0.495\linewidth}
        \centering
        \includegraphics[width=.85\linewidth]{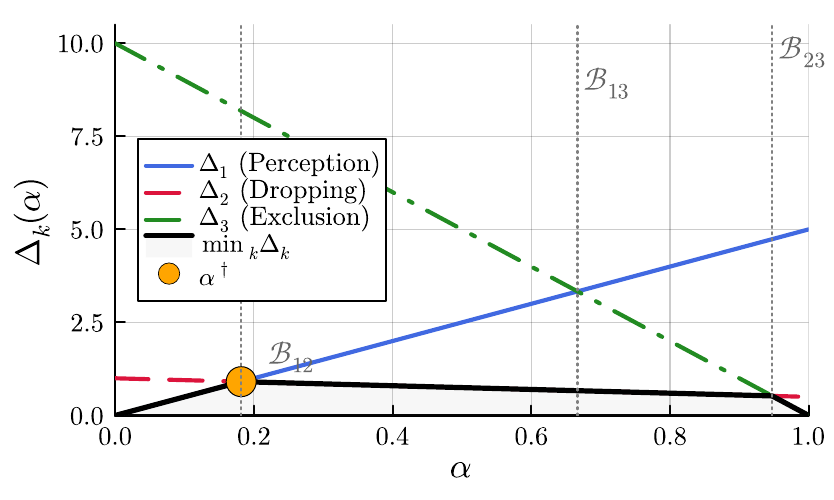}
        \caption{$\mathcal{S} = 2$ (Region~B): all three strategies active.}
        \label{fig:lower_envelope_S2}
    \end{subfigure}
    \hfill
    \begin{subfigure}[b]{0.495\linewidth}
        \centering
        \includegraphics[width=.85\linewidth]{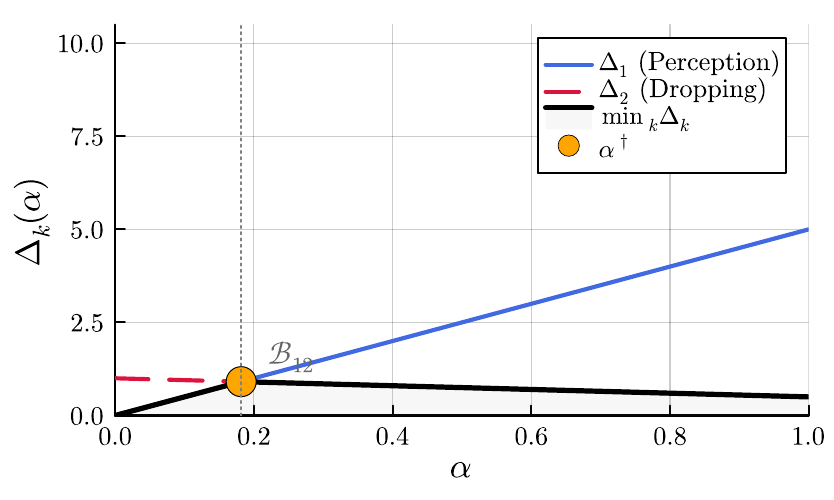}
        \caption{$\mathcal{S} = 2^{16}$ (Region~A): Exclusion dominated.}
        \label{fig:lower_envelope_S65536}
    \end{subfigure}
    \caption{Lower envelope $\min_k \Delta_k(\alpha)$ and robust operating point $\alpha^{\dagger}$ ($\varepsilon\subscript{K} = 0.5$, $\rho = 0.1$).}
    \label{fig:lower_envelope}
\end{figure*}

Given $\alpha^{\dagger}$, it remains to optimize $(n\subscript{M}, n\subscript{K})$. Sec.~\ref{sec:strategy_optimization} derives closed-form solutions for each strategy pair, assembles them into an iterative algorithm, and confirms that the iteration oscillates as predicted by Theorem~\ref{thm:best_response}.

\section{Strategy Optimization}
\label{sec:strategy_optimization}

We solve the resource-allocation and ciphering-rate sub-problems in closed form for a fixed strategy pair (Sections~\ref{sec:strategy_optimization}-A, B), then generalize to all viable pairs (Section~\ref{sec:strategy_optimization}-C). An iterative algorithm (Section~\ref{subsec:algorithm}) alternates strategy prediction with joint $(n\subscript{M}, n\subscript{K}, \alpha)$ optimization; as Theorem~\ref{thm:best_response} predicted, it oscillates (Section~\ref{subsec:oscillation}). The robust $\alpha^{\dagger}$ combined with the closed-form $(n\subscript{M}, n\subscript{K})$ bypasses the iteration entirely.

\subsection{Initial Resource Allocation Optimization}
We first maximize the distortion of \emph{Eve} for a given strategy pair, starting with the deterministic decryptor (Perception):
\begin{maxi!}
	{n\subscript{M}, n\subscript{K} \in\mathbb{Z}^+}{D\eve} {\label{obj:deceptive_distortion}}{}{}
    \addConstraint{\varepsilon\subscript{Bob,M}\leqslant \varepsilon\superscript{th}\subscript{Bob,M}, \; \varepsilon\subscript{Eve,M}\leqslant \varepsilon\superscript{th}\subscript{Eve,M} \label{con:err_message}}
	\addConstraint{\varepsilon\subscript{Bob,K}\leqslant \varepsilon\superscript{th}\subscript{Bob,K}, \; \varepsilon\subscript{Eve,K}\geqslant \varepsilon\superscript{th}\subscript{Eve,K} \label{con:err_key}}
	\addConstraint{D\bob \leqslant D\bob\superscript{th} \label{con:D_bob}}
\end{maxi!}
where $D\eve=\varepsilon_{\mathrm{Eve},\mathrm{M}}D\subscript{loss}+\alpha (1-\varepsilon_{\mathrm{Eve},\mathrm{M}})\varepsilon_{\mathrm{Eve},\mathrm{K}}D\subscript{conf}$. 


We adopt the finite-blocklength normal approximation of~\cite{polyanskiy2010channel}, under which the error probability $\varepsilon_{i,j}$ for payload $d_j$ transmitted over blocklength $n_j$ is
    $\varepsilon_{i,j} = Q \left(\sqrt{\frac{n_j}{V(\gamma_{i})}}(\mathcal{C}(\gamma_{i})-\frac{d_j}{n_j}) \ln{2}\right)$,
where  $Q(x)=\frac{1}{\sqrt{2 \pi}} \int_{x}^{\infty} e^{-t^2/2} dt$ is the Q-function in statistic, $\mathcal{C}(\gamma_i)=\log_2(1+\gamma_i)$ is the Shannon capacity, ${V}(\gamma_i)=1-\frac{1}{{{\left( 1+{{\gamma_i }} \right)}^{2}}}$ is the channel dispersion with \ac{snr} $\gamma_i=\frac{z_iP}{\sigma^2}$.

\begin{theorem}
\label{thm:D_eve_maximum}
The maximum of $D\eve$ is achieved at $(n\subscript{M}\superscript{max},n\subscript{K}\superscript{min})$ when $\alpha\geqslant \frac{D\subscript{loss}}{\varepsilon\subscript{Eve,\mathrm{K}}\superscript{max}D\subscript{conf}}$, while achieved at $(n\subscript{M}\superscript{min},n\subscript{K}\superscript{min})$ when $\alpha<\frac{D\subscript{loss}}{\varepsilon\subscript{Eve,\mathrm{K}}\superscript{max}D\subscript{conf}}$.
\end{theorem}

\noindent\emph{Proof:} See Appendix~\ref{app:proof_D_eve_monotonicity}.

The optimum always lies on the boundary: the key coding rate is maximized (to prevent Eve from decoding the key), while the ciphertext coding rate depends on $\alpha$. When $\alpha$ is large, confusion distortion dominates, so the ciphertext rate should be minimized to ensure correct decoding; when $\alpha$ is small, loss distortion dominates, and the ciphertext rate should be maximized.

To state the closed-form solution, we introduce the auxiliary function $\phi(\varepsilon_{i,j})$, which maps a target error probability to the required blocklength given payload $d_j$ and $\gamma_i$ \cite{polyanskiy2010channel}:
    \begin{equation}
        \begin{aligned}
        \phi(\varepsilon_{i,j})&=\Bigg(\frac{\log_{2}{e} \mathcal{Q}^{-1}(\varepsilon_{i,j})\sqrt{V_i}}{2\mathcal{C}_i} \\
            &+\frac{\sqrt{\left[(\log_{2}{e} \mathcal{Q}^{-1}(\varepsilon_{i,j})\right]^2V_i+4d_j\mathcal{C}_i}}{2\mathcal{C}_i}\Bigg)^2.
        \end{aligned}
    \end{equation}
The boundary blocklengths are:
    \begin{align}
    \label{eq:n_K^min}
        n\subscript{K}\superscript{min}&=\phi\!\left(\min \left\{\varepsilon\subscript{Bob,K}\superscript{th},\frac{D\bob\superscript{th}}{\alpha D\subscript{conf}}\right\}\right), \\
    \label{eq:n_M^min}
        n\subscript{M}\superscript{min}&= \min \left\{\phi\!\left(\frac{D\superscript{th}\bob}{D\subscript{loss}}\right),\phi\!\left(\varepsilon\subscript{Eve,M}\superscript{th}\right)\right\}.
    \end{align}
The maximum of $D\eve$ is then:
{\footnotesize
    \begin{equation}
    \label{eq:D_Eve_max_solution}
    D\eve\superscript{max}=
     \begin{cases}
     \varepsilon\subscript{Eve,M}\superscript{max}+ \\ \alpha \left[1-\varepsilon\subscript{Eve,M}\superscript{max}\right]\varepsilon\subscript{Eve,K}\superscript{max}D\subscript{conf}  & \text{if } \alpha < \frac{D\subscript{loss}}{\varepsilon\subscript{Eve,K}\superscript{max}D\subscript{conf}} \\
\varepsilon\subscript{Eve,M}\superscript{min}+ \\ \alpha \left[1-\varepsilon\subscript{Eve,M}\superscript{min}\right]
\varepsilon\subscript{Eve,K}\superscript{max}D\subscript{conf} &  \text{if } \alpha \geqslant \frac{D\subscript{loss}}{\varepsilon\subscript{Eve,K}\superscript{max}D\subscript{conf}}
\end{cases}
\end{equation}
}
where:
{\small
    \begin{align}
        \begin{split}            \varepsilon\subscript{Eve,K}\superscript{max}=&\mathcal{Q}\left(\ln 2\sqrt{\frac{n\subscript{K}\superscript{min}}{V\eve}}\left(\mathcal{C}\eve-\frac{d\subscript{K}}{n\subscript{K}\superscript{min}}\right) \right),
        \end{split}
        \\
        \begin{split}
        \varepsilon\subscript{Eve,M}\superscript{max}=&\mathcal{Q}\left(\ln 2 \sqrt{\frac{n\subscript{M}\superscript{min}}{V\eve}}\left(\mathcal{C}\eve-\frac{d\subscript{M}}{n\subscript{M}\superscript{min}}\right) \right),
        \end{split}
        \\
        \varepsilon\subscript{Eve,M}\superscript{min}=&\mathcal{Q}\left(\ln 2\sqrt{\frac{n\subscript{M}\superscript{max}}{V\eve}}\left(\mathcal{C}\eve-\frac{d\subscript{M}}{n\subscript{M}\superscript{max}}\right) \right).
    \end{align}
}

The auxiliary function $\phi(\cdot)$ returns real-valued blocklengths, whereas the feasibility set requires $n\subscript{M}, n\subscript{K} \in \mathbb{Z}^+$. We round $\phi(\cdot)$ up to the nearest integer in practice. At the blocklengths considered in Sec.~\ref{sec:numerical} ($n \gtrsim 100$), the unit-slot rounding perturbs $\varepsilon_{i,j}$ negligibly and leaves both the active strategy regime and the closed-form robust $\alpha^{\dagger}$ unchanged.

\subsection{Ciphering Rate Optimization}

The resource allocation solution in Sec.~\ref{sec:strategy_optimization}-A expresses $D\eve\superscript{max}$ as a function of the ciphering rate~$\alpha$; it remains to choose $\alpha$ itself. Alice knows Bob's \ac{csi} exactly and Eve's statistically, so she can predict both receivers' optimal strategies as functions of~$\alpha$. She then optimizes:
\begin{maxi!}
 {\alpha}{\min \Tilde{D}\eve(\mathbb{E}\{\varepsilon\subscript{Eve,M}\},\mathbb{E}\{\varepsilon\subscript{Eve,K}\})} {\label{Problem_min_Tilde_D_Eve_w.r.t.alpha}}{\label{obj:min_D_Eve}}{}
	\addConstraint{\min \Tilde{D}\bob \leqslant \Tilde{D}\superscript{th}\bob},
\end{maxi!}
Both the objective and constraint are linear in~$\alpha$, so the optimum lies on the boundary.
When $\beta\superscript{Eve}_1=1$, $\min \Tilde{D}\eve$ increases in~$\alpha$, giving:
\begin{equation}
\label{eq:alpha^o_beta_1}
    \alpha\superscript{\mathrm{o}}=\begin{cases}
    \min \left\{\frac{\Tilde{D}\superscript{th}\bob-\varepsilon\subscript{Bob,M}D\subscript{loss}}{(1-\varepsilon\subscript{Bob,M})\varepsilon\subscript{Bob,K}D\subscript{conf}},1\right\} & \text{if } \beta_1\superscript{Bob}=1  \\ 1 & \text{if } \beta_2\superscript{Bob}=1\\
    1 & \text{if } \beta_3\superscript{Bob}=1.
    \end{cases}
\end{equation}

When $\beta\superscript{Eve}_2=1$ or $\beta\superscript{Eve}_3=1$, $\min \Tilde{D}\eve$ decreases in~$\alpha$, giving:
{\scriptsize
\begin{equation}
\label{eq:alpha^o_beta_2and3}
 \alpha\superscript{\mathrm{o}}=   \begin{cases}
0  & \text{if } \beta_1\superscript{Bob}=1 \\
 \frac{D\subscript{loss}-\Tilde{D}\superscript{th}\bob}{(1-\varepsilon\subscript{Bob,K})(1-\varepsilon\subscript{Bob,M})D\subscript{loss}}  & \text{if } \beta_2\superscript{Bob}=1 \\
\max\left\{\frac{(\mathcal{S}-1)}{\left[(\mathcal{S}-2)\varepsilon\subscript{Bob,K}-(\mathcal{S}-1)\right]}\cdot \right.\\
\left.\frac{\left[\Tilde{D}\superscript{th}\bob-\varepsilon\subscript{Bob,M}D\subscript{loss}-(1-\varepsilon\subscript{Bob,M})D\subscript{conf}\right]}{(1-\varepsilon\subscript{Bob,M})D\subscript{conf}},0\right\}  & \text{if } \beta_3\superscript{Bob}=1.
\end{cases}
\end{equation}}

\subsection{Adaptive Resource Allocation Optimization}
The preceding two subsections solved the ciphering-rate and resource-allocation sub-problems for a single strategy pair. Since the iterative algorithm (Sec.~\ref{subsec:algorithm}) re-predicts strategies each iteration, we need the optimal boundary $(n\subscript{M}, n\subscript{K})$ for every viable (Bob, Eve) combination. Given $\alpha\superscript{\mathrm{o}}$, we optimize:
\begin{maxi!}
 {n\subscript{M},n\subscript{K}}{\min \Tilde{D}\eve(\mathbb{E}\{\varepsilon\subscript{Eve,M}\},\mathbb{E}\{\varepsilon\subscript{Eve,K}\},\alpha^{\mathrm{o}})} {\label{Problem_min_Tilde_D_Eve_w.r.t.nMnK}}{\label{obj:max_min_Tilde_D_Eve}}{}
	\addConstraint{\eqref{con:err_message}, \eqref{con:err_key} \label{con:error_probability}}
	\addConstraint{\min \Tilde{D}\bob (\alpha^{\mathrm{o}}) \leqslant \Tilde{D}\bob\superscript{th} \label{con:min_Tilde_D_bob}}.
\end{maxi!}

The solution depends on the (Bob, Eve) strategy pair. In principle, all $3 \times 3 = 9$ combinations must be considered. However, Lemma~\ref{lem:impossible_pairs} eliminates those pairs where Bob adopts a more aggressive strategy than Eve: in Region~A, only $(P,P)$, $(P,D)$, and $(D,D)$ are viable; in Region~B, six of the nine pairs remain. We derive the boundary conditions for each of Bob’s three strategies below; the corresponding Eve-side optimality analysis follows.

\subsubsection{Boundary for Bob's Perception Strategy}

When $\beta_1\superscript{Bob}=1$, the distortion $\Tilde{D}\bob^{\beta_1}=\varepsilon\subscript{Bob,M}D\subscript{loss}+\alpha^{\mathrm{o}}(1-\varepsilon\subscript{Bob,M})\varepsilon\subscript{Bob,K}D\subscript{conf}$ coincides with~\eqref{con:D_bob}, simplifying the constraints to:
\begin{equation}
    \varepsilon\subscript{Bob,K}\leqslant\frac{\Tilde{D}\superscript{th}\subscript{Bob}}{\alpha^{\mathrm{o}} D\subscript{conf}}, \quad
    \varepsilon\subscript{Bob,M}\leqslant\frac{\Tilde{D}\bob\superscript{th}}{D\subscript{loss}}.
\end{equation}

The boundary blocklengths are:
\begin{align}
\begin{split}
    n\subscript{M}\superscript{min}=\min \left\{\phi\left(\varepsilon\subscript{Eve,M}\superscript{th}\right),\phi\left(\frac{\Tilde{D}\bob\superscript{th}}{D\subscript{loss}}\right)\right\}
    \end{split} \\
    \begin{split}
    n\subscript{K}\superscript{min}=\phi\left[\min \left(\varepsilon\subscript{Bob,K}\superscript{th},\frac{\Tilde{D}\superscript{th}\subscript{Bob}}{\alpha^{\mathrm{o}} D\subscript{conf}}\right)\right].
    \end{split}
\end{align}

\subsubsection{Boundary for Bob's Dropping Strategy}
When $\beta_2\superscript{Bob}=1$, $\Tilde{D}\bob^{\beta_2}=\varepsilon\subscript{Bob,M}D\subscript{loss}+(1-\varepsilon\subscript{Bob,M})\left[\varepsilon\subscript{Bob,K}\alpha^{\mathrm{o}}+(1-\alpha^{\mathrm{o}})\right]D\subscript{loss}$. Since $\Tilde{D}\bob\superscript{th}<D\subscript{loss}$, the joint bound from~\eqref{con:min_Tilde_D_bob} decreases in each error probability; evaluating at the boundary yields the simplified bounds
    \begin{align}
        &\varepsilon\subscript{Bob,K}\leqslant\frac{\Tilde{D}\bob\superscript{th}+(\alpha^{\mathrm{o}}-1)D\subscript{loss}}{\alpha^{\mathrm{o}}D\subscript{loss}},\\
        &\varepsilon\subscript{Bob,M}\leqslant 1-\frac{D\subscript{loss}-\Tilde{D}\superscript{th}\bob}{\alpha^{\mathrm{o}}D\subscript{loss}}.
    \end{align}
    The boundary blocklengths are:
    \begin{align}
        &n\subscript{M}\superscript{min}=\min \left\{\phi\left(\varepsilon\subscript{Eve,M}\superscript{th}\right),\phi\left(1-\frac{D\subscript{loss}-\Tilde{D}\superscript{th}\bob}{\alpha^{\mathrm{o}}D\subscript{loss}}\right)\right\},\\
        &n\subscript{K}\superscript{min}=\phi\left[\min \left(\varepsilon\subscript{Bob,K}\superscript{th},\frac{\Tilde{D}\bob\superscript{th}+(\alpha^{\mathrm{o}}-1)D\subscript{loss}}{\alpha^{\mathrm{o}}D\subscript{loss}}\right)\right].
    \end{align}

\subsubsection{Boundary for Bob's Exclusion Strategy}
When $\beta_3\superscript{Bob}=1$, $\Tilde{D}\bob^{\beta_3}=\varepsilon\subscript{Bob,M}D\subscript{loss}+(1-\varepsilon\subscript{Bob,M})\left[\varepsilon\subscript{Bob,K}\frac{\alpha^{\mathrm{o}}(\mathcal{S}-2)}{\mathcal{S}-1}D\subscript{conf}+(1-\alpha^{\mathrm{o}})D\subscript{conf}\right]$. By the same monotonicity argument (using $\Tilde{D}\bob\superscript{th}<D\subscript{loss}(\mathcal{S}-1)$), the simplified bounds are
\begin{align}
    &\varepsilon\subscript{Bob,K}\leqslant\frac{\Tilde{D}\superscript{th}\bob-(1-\alpha^{\mathrm{o}})(\mathcal{S}-1)D\subscript{conf}}{\alpha^{\mathrm{o}}(\mathcal{S}-2)D\subscript{conf}},\\
    &\varepsilon\subscript{Bob,M}\leqslant \frac{\Tilde{D}\bob\superscript{th}-(1-\alpha^{\mathrm{o}})D\subscript{conf}}{D\subscript{loss}-(1-\alpha^{\mathrm{o}})D\subscript{conf}}.
\end{align}

Since $D\eve$ is always maximized at $n\subscript{K}\superscript{min}$ and $n\subscript{M}\superscript{max}$ is set by practical requirements, the boundary blocklengths for all three strategies are:

{\small
\begin{equation}
\label{eq:n_M^min_op}
n\subscript{M}\superscript{min}=
\begin{cases}
 \min \left\{\phi\left(\varepsilon\subscript{Eve,M}\superscript{th}\right),\phi\left(\frac{\Tilde{D}\bob\superscript{th}}{D\subscript{loss}}\right)\right\}   & \text{if } \beta_1\superscript{Bob}=1\\
 \begin{aligned}
&\min  \Bigg\{\phi\left(\varepsilon\subscript{Eve,M}\superscript{th}\right), \\
&\phi\left(1-\frac{D\subscript{loss}-\Tilde{D}\superscript{th}\bob}{\alpha^{\mathrm{o}}D\subscript{loss}}\right)\Bigg\} \end{aligned}   & \text{if } \beta_2\superscript{Bob}=1\\
\begin{aligned}
 &\min \Bigg\{\phi\left(\varepsilon\subscript{Eve,M}\superscript{th}\right), \\
 &\phi\left(\frac{\Tilde{D}\bob\superscript{th}-(1-\alpha^{\mathrm{o}})D\subscript{conf}}{D\subscript{loss}-(1-\alpha^{\mathrm{o}})D\subscript{conf}}\right)\Bigg\} \end{aligned}   & \text{if } \beta_3\superscript{Bob}=1,
\end{cases}
\end{equation}}

{\footnotesize
\begin{equation}
\label{eq:n_K^min_op}
n\subscript{K}\superscript{min}=
\begin{cases}
n\subscript{K}\superscript{min}=\phi\left[\min \left(\varepsilon\subscript{Bob,K}\superscript{th},\frac{\Tilde{D}\superscript{th}\subscript{Bob}}{\alpha^{\mathrm{o}} D\subscript{conf}}\right)\right]  & \text{if } \beta_1\superscript{Bob}=1\\
\begin{aligned}
 &\phi\Bigg[\min \bigg(\varepsilon\subscript{Bob,K}\superscript{th}, \\
 &\frac{\Tilde{D}\bob\superscript{th}+(\alpha^{\mathrm{o}}-1)D\subscript{loss}}{\alpha^{\mathrm{o}}D\subscript{loss}}\bigg)\Bigg] \end{aligned}  & \text{if } \beta_2\superscript{Bob}=1 \\
 \begin{aligned}
&\phi \Bigg[\min \bigg(\varepsilon\subscript{Bob,K}\superscript{th}, \\
&\frac{\Tilde{D}\superscript{th}\bob-(1-\alpha^{\mathrm{o}})(\mathcal{S}-1)D\subscript{conf}}{\alpha^{\mathrm{o}}(\mathcal{S}-2)D\subscript{conf}}\bigg)\Bigg]  \end{aligned}   & \text{if } \beta_3\superscript{Bob}=1.
\end{cases}
\end{equation}}

We now determine the optimal boundary point for each of Eve's strategies.

\subsubsection{Optimal Point for Eve's Perception Strategy}
When $\beta_1\superscript{Eve}=1$, the objective reduces to~\eqref{eq:D_i_resource_allocation}, so the solution follows directly from~\eqref{eq:D_Eve_max_solution}.

\subsubsection{Optimal Point for Eve's Dropping Strategy}
When $\beta_2\superscript{Eve}=1$, the objective becomes:
\begin{equation}
\begin{split}
    \Tilde{D}\eve^{\beta_2} &=\varepsilon\subscript{Eve,M}D\subscript{loss} \\
    &+(1-\varepsilon\subscript{Eve,M})\left[\varepsilon\subscript{Eve,K}\alpha^{\mathrm{o}}+(1-\alpha^{\mathrm{o}})\right]D\subscript{loss}.
    \end{split}
\end{equation}

\begin{theorem}
\label{thm:dropping_monotonicity}
    If $\beta_2\superscript{Eve}=1$, $\Tilde{D}\eve^{\beta_2}$ is maximized at $(n\subscript{M}\superscript{min},n\subscript{K}\superscript{min})$.
\end{theorem}
\noindent\emph{Proof:} See Appendix~\ref{app:proof_dropping_monotonicity}. 

The maximum is:
    {\small
    \begin{equation}
        \begin{split}
           \max \Tilde{D}\eve^{\beta_2}&
           =\varepsilon\subscript{Eve,M}\superscript{max}D\subscript{loss} \\
    &+(1-\varepsilon\subscript{Eve,M}\superscript{max})\left[\varepsilon\subscript{Eve,K}\superscript{max}\alpha^{\mathrm{o}}+(1-\alpha^{\mathrm{o}})\right]D\subscript{loss},
        \end{split}
    \end{equation}}
    where \begin{align}
        \varepsilon\subscript{Eve,M}\superscript{max}=\mathcal{Q}\left(\sqrt{\frac{n\subscript{M}\superscript{min}}{V\eve}}\left(\mathcal{C}\eve-\frac{d\subscript{M}}{n\subscript{M}\superscript{min}}\right)\ln 2 \right),\\
        \varepsilon\subscript{Eve,K}\superscript{max}=\mathcal{Q}\left(\sqrt{\frac{n\subscript{K}\superscript{min}}{V\eve}}\left(\mathcal{C}\eve-\frac{d\subscript{K}}{n\subscript{K}\superscript{min}}\right)\ln 2 \right).
    \end{align}
    
    \subsubsection{Optimal Point for Eve's Exclusion Strategy}
When $\beta_3\superscript{Eve}=1$, the objective becomes:
\begin{align}
    \begin{split}
    \Tilde{D}^{\beta_3}\subscript{Eve}&=\varepsilon\subscript{Eve,M}D\subscript{loss} \\
    &+(1-\varepsilon\subscript{Eve,M}) \cdot \\
    &\left[\varepsilon\subscript{Eve,K}\frac{\alpha^{\mathrm{o}}(\mathcal{S}-2)}{\mathcal{S}-1}D\subscript{conf}+(1-\alpha^{\mathrm{o}})D\subscript{conf}\right].
    \end{split}
    \end{align}
    \begin{theorem}
    \label{thm:exclusion_monotonicity}
        If $\beta_3\superscript{Eve}=1$, the maximum of $\Tilde{D}\eve^{\beta_3}$ is obtained at $(n\subscript{M}\superscript{max},n\subscript{K}\superscript{min})$ when $\alpha^{\mathrm{o}}<\frac{D\subscript{conf}-D\subscript{loss}}{D\subscript{conf}(1-\varepsilon\subscript{Eve,K}\superscript{max}\frac{\mathcal{S}-2}{\mathcal{S}-1})}$, while achieved at $(n\subscript{M}\superscript{min},n\subscript{K}\superscript{min})$ when $\alpha^{\mathrm{o}}\geqslant\frac{D\subscript{conf}-D\subscript{loss}}{D\subscript{conf}(1-\varepsilon\subscript{Eve,K}\superscript{max}\frac{\mathcal{S}-2}{\mathcal{S}-1})}$.
    \end{theorem}
    \noindent\emph{Proof:} See Appendix~\ref{app:proof_exclusion_monotonicity}.

\subsection{Optimization Algorithm}\label{subsec:algorithm}
Sections~\ref{sec:strategy_optimization}-A through~C solved the optimization for a given strategy pair, but the pair itself depends on the operating parameters, creating a circular dependency. Alg.~\ref{alg:optimization} resolves this by iterating: at each step Alice predicts the receivers' strategies from the current parameters and then re-optimizes.

\begin{algorithm}[!htbp]
\scriptsize
\SetAlgoLined
	Input: $T$, $\Tilde{D}\superscript{th}\bob$, $\varepsilon_{\mathrm{Bob,M}}$, $\varepsilon_{\mathrm{Bob,K}}$, $\mathbb{E}\{\varepsilon_{\mathrm{Eve,M}}\}$,$\mathbb{E}\{\varepsilon_{\mathrm{Eve,K}}\}$  \\
	Initialize: $n\subscript{M}=n\subscript{M}\superscript{init}$, $n\subscript{K}=n\subscript{K}\superscript{init}$, $\alpha=\alpha\superscript{init}$, $\beta_1\superscript{Bob}=1$, $\beta_2\superscript{Bob}=0$, $\beta_3\superscript{Bob}=0$, $\beta_1\superscript{Eve}=1$, $\beta_2\superscript{Eve}=0$, $\beta_3\superscript{Eve}=0$, \\
    where
    $n\subscript{M}\superscript{init}=\min \left\{\phi\left(\varepsilon\subscript{Eve,M}\superscript{th}\right),\phi\left(\frac{\Tilde{D}\bob\superscript{th}}{D\subscript{loss}}\right)\right\}$,
    $n\subscript{K}\superscript{init}=\phi\left[\min \left(\varepsilon\subscript{Bob,K}\superscript{th},\frac{\Tilde{D}\superscript{th}\subscript{Bob}}{\alpha\superscript{init} D\subscript{conf}}\right)\right]$ \\
    \Do{$t\leqslant T$}
        {
        
        $(\beta_1^{\mathrm{Eve}(t)},\beta_2^{\mathrm{Eve}(t)}, \beta_3^{\mathrm{Eve}(t)}) \leftarrow \arg \underset{\beta_1,\beta_2, \beta_3}{\min} \Tilde{D}\eve(n\subscript{M}^{(t-1)},n\subscript{K}^{(t-1)})$, \\
        $(\beta_1^{\mathrm{Bob}(t)},\beta_2^{\mathrm{Bob}(t)}, \beta_3^{\mathrm{Bob}(t)}) \leftarrow \arg \underset{\beta_1,\beta_2, \beta_3}{\min} \Tilde{D}\bob(n\subscript{M}^{(t-1)},n\subscript{K}^{(t-1)})$ \\
        $\min \Tilde{D}\eve^{(t)}:=\Tilde{D}\eve(\beta_1^{\mathrm{Eve}(t)},\beta_2^{\mathrm{Eve}(t)}, \beta_3^{\mathrm{Eve}(t)}) $\\
        $\min \Tilde{D}\bob^{(t)}:=\Tilde{D}\bob(\beta_1^{\mathrm{Bob}(t)},\beta_2^{\mathrm{Bob}(t)}, \beta_3^{\mathrm{Bob}(t)}) $\\
        $\alpha^{\mathrm{o}(t)} \leftarrow\arg \underset{\alpha}{\max} \left[\min \Tilde{D}\eve^{(t)}\right]$ \\
$(n\subscript{M}^{(t)},n\subscript{K}^{(t)})\leftarrow\arg \underset{n\subscript{M},n\subscript{K}}{\max} \left[\min \Tilde{D}\eve^{(t)}(\alpha^{\mathrm{o}(t)})\right]$ \\
$t\leftarrow t+1$} 
        \textbf{return} $n\subscript{M}$, $n\subscript{K}$ and $\alpha^{\mathrm{o}}$
        \caption{Optimization of $n\subscript{M},n\subscript{K}$ and $\alpha^{\mathrm{o}}$}
        \label{alg:optimization}
\end{algorithm}

Alg.~\ref{alg:optimization} initializes with both receivers adopting Perception; simulations confirm the initial strategy choice does not affect the subsequent periodic behavior. At each iteration, Alice predicts both receivers’ strategies via~\eqref{eq:Tilde_D_i}, optimizes $\alpha$ via~\eqref{eq:alpha^o_beta_1}--\eqref{eq:alpha^o_beta_2and3}, and re-allocates $(n\subscript{M}, n\subscript{K})$ via~\eqref{eq:n_M^min_op}--\eqref{eq:n_K^min_op}. Since the predicted strategies may change each iteration, convergence is not guaranteed; the loop terminates after $T$ iterations. All per-iteration updates are closed-form, so the complexity is $\mathcal{O}(T)$, lower than the deception-rate formulation in~\cite{chen2025physical}.

\subsection{Oscillation Analysis and Robust Alternative}
\label{subsec:oscillation}

A natural question is whether Alg.~\ref{alg:optimization} converges to a stable operating point. Unfortunately, it does not: as Sec.~\ref{sec:numerical} will reveal, the algorithm exhibits periodic oscillation, and Theorem~\ref{thm:best_response} established that the resulting time-average is strictly suboptimal relative to~$\alpha^{\dagger}$. We now trace the oscillation to a structural property of the \ac{lp} formulation, which also points to the fix.

\begin{proposition}[Oscillation Characterization]
\label{prop:oscillation}
In the typical \ac{pls} operating regime where the channel gap $\gamma\subscript{Bob} - \gamma\subscript{Eve}$ is sufficiently large, Alg.~\ref{alg:optimization} produces periodic oscillation with period~2. The mechanism is:
\begin{enumerate}
    \item Assuming Eve adopts Perception: Alice maximizes $\alpha$ via \eqref{eq:alpha^o_beta_1}, pushing $\alpha\superscript{o}$ above the switching threshold $\alpha_{12}^*(\varepsilon\subscript{Eve,K})$ from \eqref{eq:B12}.
    \item Eve's predicted best response switches to a more defensive strategy (Dropping or Exclusion).
    \item Assuming Eve adopts the defensive strategy: Alice minimizes $\alpha$ via \eqref{eq:alpha^o_beta_2and3}, pushing $\alpha\superscript{o}$ below $\alpha_{12}^*(\varepsilon\subscript{Eve,K})$.
    \item Eve's predicted best response switches back to Perception.
\end{enumerate}
\end{proposition}

\noindent\emph{Proof:} See Appendix~\ref{app:proof_oscillation}.

\begin{remark}
The condition ``sufficiently large channel gap'' means that $\alpha\superscript{o}$ from \eqref{eq:alpha^o_beta_1} exceeds $\alpha_{12}^*(\varepsilon\subscript{Eve,K})$. When $\gamma\subscript{Bob} - \gamma\subscript{Eve}$ is small, $\varepsilon\subscript{Bob,K}$ is large, $\alpha\superscript{o}$ remains below the switching threshold, and the algorithm converges in one step to a Perception-only equilibrium. In typical \ac{pls} deployments, Alice has a substantial channel advantage, so the oscillation regime is the practically relevant case.
\end{remark}

The root cause is that the \ac{lp} boundary solutions~\eqref{eq:alpha^o_beta_1} and~\eqref{eq:alpha^o_beta_2and3} always push $\alpha$ to an extreme, overshooting the switching surface in both directions. The robust operating point $\alpha^{\dagger}$ (Theorem~\ref{thm:robust_alpha}), by contrast, sits precisely \emph{at} the switching boundary and avoids oscillation entirely.

\textbf{Robust alternative.} Alice directly computes $\alpha^{\dagger}$ from~\eqref{eq:robust_alpha} and optimizes $(n\subscript{M}, n\subscript{K})$ for that fixed $\alpha^{\dagger}$, yielding a provably higher worst-case guarantee $\Tilde{D}\eve\superscript{wc}$ than the iterative algorithm (Fig.~\ref{fig:oscillation_traces}).

\section{Numerical Evaluation}
\label{sec:numerical}
We verify three aspects: \begin{enumerate*}[label=\emph{\roman*)}]
    \item the optimization landscape matches the predicted structure, \item the robust design outperforms the iterative algorithm, and \item the advantage extends to fading channels. Table~\ref{tab:sim_setup} lists the simulation parameters; the distortion ratio $\rho = D\subscript{loss}/D\subscript{conf} = 0.1$ reflects that undetected deception is ten times costlier than message loss, and the qualitative regime behavior depends only on~$\rho$.
\end{enumerate*}
\begin{table}[!htpb]
	\centering
	\caption{Simulation setup}
	\label{tab:sim_setup}
	\begin{tabular}{m{2.1cm} m{1cm} m{4.5cm}}
		\toprule[2px]
		\textbf{Parameter} 		&\textbf{Value} 		&\textbf{Remark}\\
		\midrule[1px]
		$\sigma^2$					&\SI{1}{\milli\watt}			&Noise power\\
  \rowcolor{gray!30}
		$B$								 &\SI{1}{\hertz}			&Normalized to unity bandwidth\\
		$d\subscript{M}$								 &16 bits							 &Length of ciphertext\\
  \rowcolor{gray!30}
        $d\subscript{K}$								 &16 bits							 &Length of key\\
        $D\subscript{loss}$        & 1                  & Distortion between plaintext and ciphertext\\
        \rowcolor{gray!30} $D\subscript{conf}$        & 10                  & Distortion between plaintext and invalid decoding codeword\\

		$\varepsilon\subscript{Bob,M}\superscript{th}$, $\varepsilon\subscript{Bob,K}\superscript{th}$
        $\varepsilon\subscript{Eve,M}\superscript{th}$, $\varepsilon\subscript{Eve,K}\superscript{th}$		&0.5							 &Thresholds in constraints \eqref{con:err_message}, \eqref{con:err_key}\\
        \rowcolor{gray!30}
        $\Tilde{D}\bob\superscript{th}$  &0.01 &Threshold in constraints \eqref{con:min_Tilde_D_bob} \\
        T  &  100  &  Maximum iterations of simulation\\
		\bottomrule[2px]
	\end{tabular}
\end{table}

\subsection{Semantic Distortion Surface}
\begin{figure}[!htpb]
	\centering
	\includegraphics[width=.8\linewidth, clip, trim=60 0 55 40]{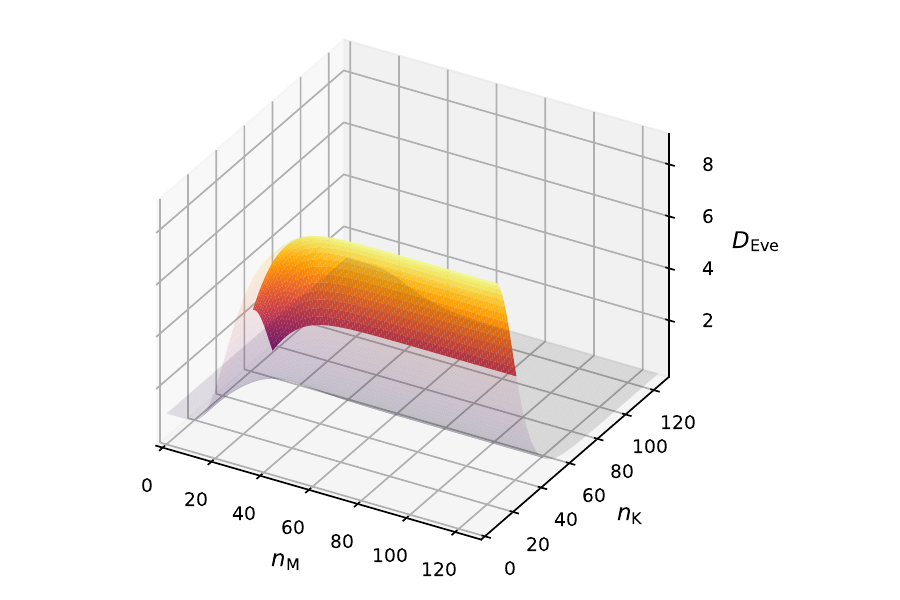}
	\caption{Semantic distortion with $\alpha=0.9$.}
	\label{fig:deception_rate-surface}
\end{figure}
We first verify that the resource allocation subproblem has the concavity structure exploited by the closed-form solution in Sec.~\ref{sec:strategy_optimization}.
Fig.~\ref{fig:deception_rate-surface} plots $\min \Tilde{D}\eve$ over $(n\subscript{M},n\subscript{K}) \in \{1,\ldots,128\}^2$ for the representative Perception--Perception case ($z\bob=\SI{0}{\dB}$, $z\eve=\SI{-10}{\dB}$, $\alpha=0.9$). The feasible region (higher opacity) confirms the concavity and monotonicity predicted by the analysis for $\alpha\geqslant D\subscript{loss}/(\varepsilon\subscript{Eve,K}\superscript{max}D\subscript{conf})$.

\subsection{Validation of the Proposed Algorithm}
\begin{figure}[!htpb]
    \centering
    \includegraphics[width=\linewidth]{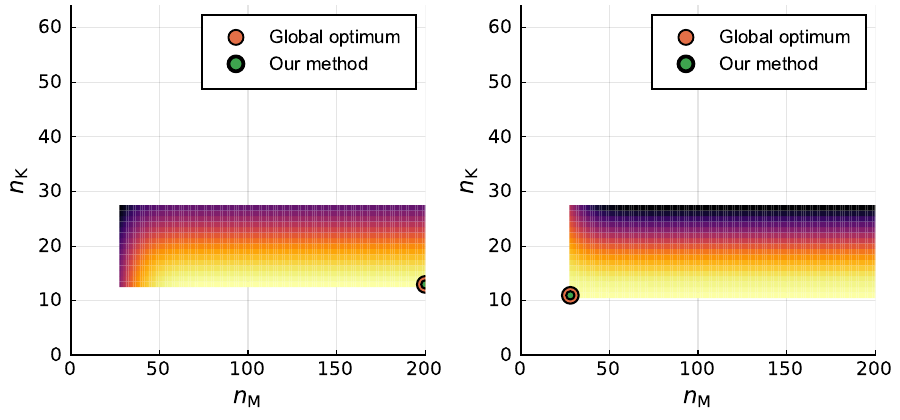}
    \caption{Global optimum in the feasible region with $\alpha=0.9$ (left) and $\alpha=0.1$ (right).}
    \label{fig:Our_method_semantic_distortion_demo}
\end{figure}

To validate that the boundary-based algorithm identifies the global optimum, we compare two ciphering rates on opposite sides of the monotonicity threshold.
Fig.~\ref{fig:Our_method_semantic_distortion_demo} compares $\alpha = 0.9$ (left) and $\alpha = 0.1$ (right) under the same Perception--Perception setting. The monotonicity of $D\eve$ with respect to $n\subscript{M}$ reverses across the threshold $\alpha = D\subscript{loss}/(\varepsilon\subscript{Eve,K}\superscript{max}D\subscript{conf})$, confirming the analysis. In both cases the feasible region is rectangular, so the global optimum lies at a corner and is found in closed form.

\subsection{Optimization of Resource Allocation and Encryption Rate}
We now test whether the iterative alternation between strategy prediction and joint optimization converges. We run Alg.~\ref{alg:optimization} for $T = 100$ iterations with $z\bob = \SI{5}{\dB}$, $z\eve = \SI{0}{\dB}$, initializing at $\alpha = 0.9$ with both receivers in Perception. Both Bob and Eve know Alice’s parameters and select the strategy minimizing their individual distortion.

\begin{figure*}[!htpb]
    \centering
    \begin{subfigure}[b]{.495\linewidth}
        \centering
        \includegraphics[width=.9\linewidth]{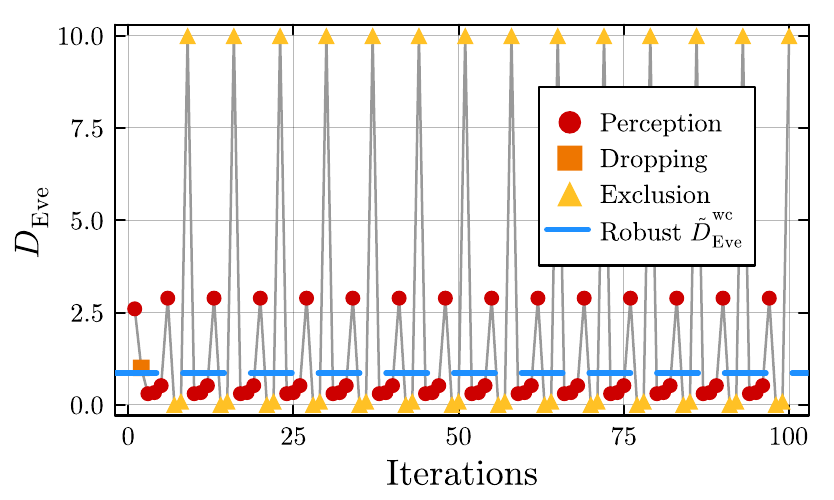}
        \caption{$D\eve$, $\mathcal{S}=2$}
        \label{fig:Deve_S=2}
    \end{subfigure}
    \hfill
    \begin{subfigure}[b]{.495\linewidth}
        \centering
        \includegraphics[width=.9\linewidth]{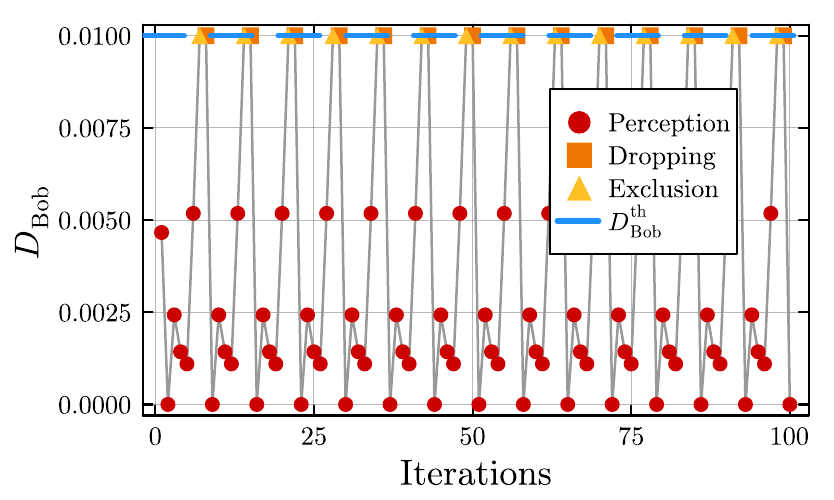}
        \caption{$D\bob$, $\mathcal{S}=2$}
        \label{fig:Dbob_S=2}
    \end{subfigure}
    \\[6pt]
    \begin{subfigure}[b]{.495\linewidth}
        \centering
        \includegraphics[width=.9\linewidth]{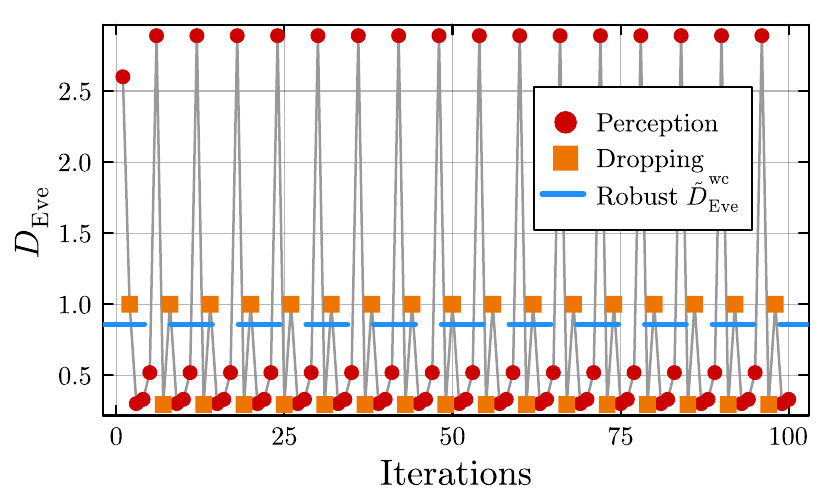}
        \caption{$D\eve$, $\mathcal{S}=2^{16}$}
        \label{fig:Deve_S=2^16}
    \end{subfigure}
    \hfill
    \begin{subfigure}[b]{.495\linewidth}
        \centering
        \includegraphics[width=.9\linewidth]{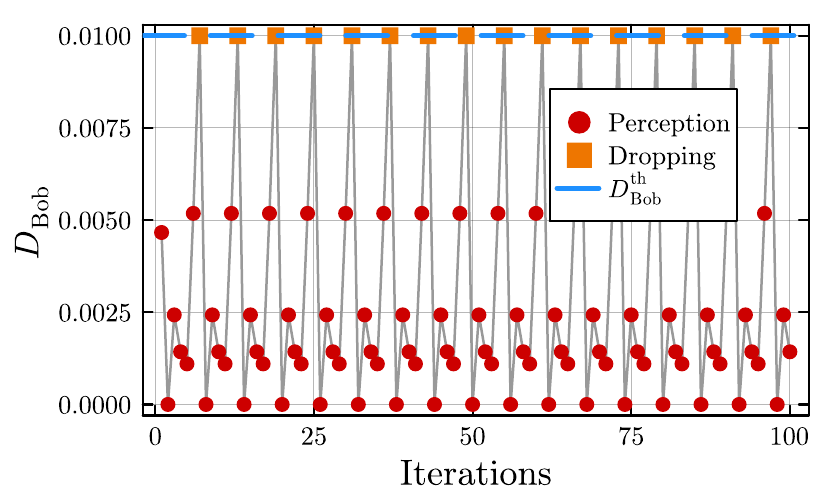}
        \caption{$D\bob$, $\mathcal{S}=2^{16}$}
        \label{fig:Dbob_S=2^16}
    \end{subfigure}
    \caption{Iterative distortion trajectories under dynamic strategy updates compared to the robust operating point.}
    \label{fig:oscillation_traces}
\end{figure*}

Fig.~\ref{fig:oscillation_traces} shows the resulting trajectories. Distortion values oscillate periodically rather than converging, consistent with Theorem~\ref{thm:best_response}. For $\mathcal{S}=2$ (top), Eve favors Exclusion due to the small codebook while Bob predominantly selects Perception. For $\mathcal{S}=2^{16}$ (bottom), neither player selects Exclusion, but the periodic pattern persists. In both cases the robust minimax point (dashed) exceeds the iterative time-average.

\subsection{Robust vs.\ Iterative Comparison}\label{subsec:robust_comparison}

Having observed periodic oscillation in the iterative dynamics, we now systematically compare the robust and iterative approaches across eavesdropper channel conditions.
Fig.~\ref{fig:minimax_comparison} sweeps $z\eve$ for both codebook sizes. The robust approach consistently achieves $\Tilde{D}\eve \approx 1$, while the iterative worst case drops to $\approx 0.5$, roughly half the robust guarantee. Both degrade as Eve's channel improves, but the robust solution maintains its advantage throughout.

\begin{figure*}[!htpb]
    \centering
    \begin{subfigure}[b]{.495\linewidth}
        \centering
        \includegraphics[width=.9\linewidth]{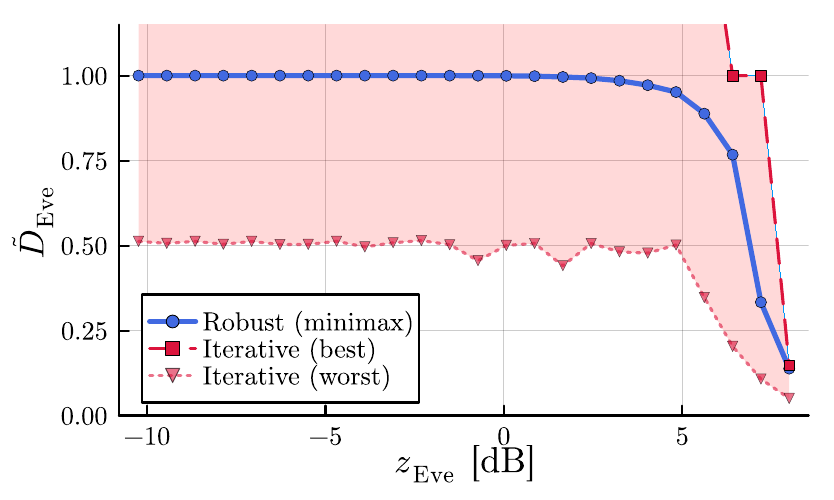}
        \caption{$\mathcal{S} = 2$ (Region~B)}
        \label{fig:comparison_S2}
    \end{subfigure}
    \hfill
    \begin{subfigure}[b]{.495\linewidth}
        \centering
        \includegraphics[width=.9\linewidth]{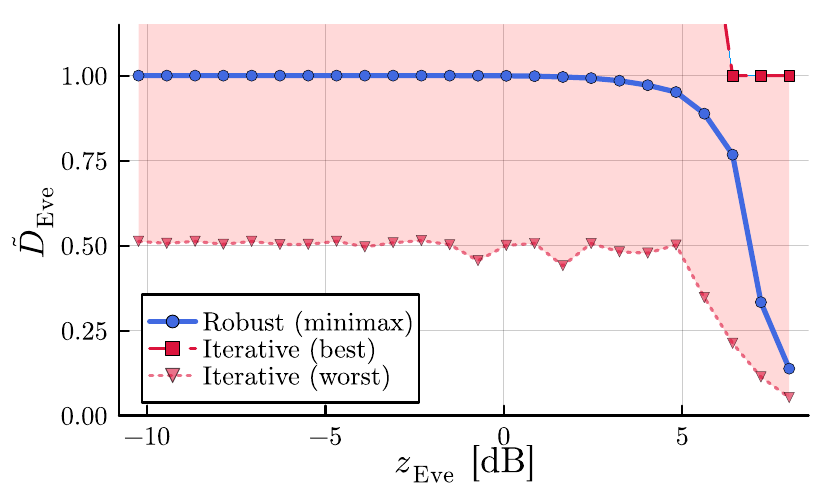}
        \caption{$\mathcal{S} = 2^{16}$ (Region~A)}
        \label{fig:comparison_S65536}
    \end{subfigure}
    \caption{Robust $\alpha^{\dagger}$ (Theorem~\ref{thm:robust_alpha}) vs.\ Alg.~\ref{alg:optimization} ($T=30$, $z\bob = \SI{10}{\dB}$). Shaded band: iterative oscillation range.}
    \label{fig:minimax_comparison}
\end{figure*}

\subsection{Fading Channel Evaluation}\label{subsec:fading}

The preceding analysis assumed \ac{awgn}; we now ask whether the deception advantage survives under fading. The instantaneous \ac{snr} is $\gamma_i = |h_i|^2 P / \sigma^2$ where $|h_i|^2 \sim \mathrm{Gamma}(m, \Omega_i)$, and three fading depths are considered: no fading ($m\to\infty$), shallow ($m=5$, Rician-like), and Rayleigh ($m=1$).

Two transmitter strategies are compared: \emph{Static} (design once from mean \acp{snr}) and \emph{Adaptive} (re-optimize per coherence block using instantaneous $\gamma\bob$ and statistical $\bar{\gamma}\eve$). As a baseline, a classical \ac{pls} scheme sets $\alpha = 0$ so that the only security comes from Eve's packet error rate ($\Tilde{D}_i = \varepsilon_{i,\mathrm{M}} D\subscript{loss}$). Table~\ref{tab:fading_evaluation} reports the results.

\begin{table}[!htpb]
\centering
\caption{Average semantic distortion under Nakagami-$m$ fading ($z\bob = \SI{5}{\dB}$, $z\eve = \SI{0}{\dB}$, $N = 10~000$).}
\label{tab:fading_evaluation}
\begin{tabular}{c c c c c c c}
\toprule[2px]
& & & \multicolumn{2}{c}{\textbf{Static}} & \multicolumn{2}{c}{\textbf{Adaptive}} \\
\cmidrule(lr){4-5} \cmidrule(lr){6-7}
$|\mathcal{S}|$ & Scheme & Fading & $\bar{\Tilde{D}}\eve$ & $\bar{\Tilde{D}}\bob$ & $\bar{\Tilde{D}}\eve$ & $\bar{\Tilde{D}}\bob$ \\
\midrule[1px]
\multirow{6}{*}{$2$} & \multirow{3}{*}{\ac{pld}} & No fading & 0.860 & 0.002 & 0.860 & 0.002 \\
 &  & Shallow & 0.715 & 0.077 & 0.643 & 0.005 \\
 &  & Rayleigh & 0.707 & 0.329 & 0.551 & 0.017 \\
\cmidrule(lr){2-7}
 & \multirow{3}{*}{\ac{pls}} & No fading & 0.556 & 0.008 & 0.556 & 0.008 \\
 &  & Shallow & 0.587 & 0.073 & 0.457 & 0.004 \\
 &  & Rayleigh & 0.648 & 0.298 & 0.485 & 0.007 \\
\midrule
\multirow{6}{*}{$2^{16}$} & \multirow{3}{*}{\ac{pld}} & No fading & 0.860 & 0.002 & 0.860 & 0.002 \\
 &  & Shallow & 0.721 & 0.077 & 0.644 & 0.005 \\
 &  & Rayleigh & 0.707 & 0.329 & 0.575 & 0.034 \\
\cmidrule(lr){2-7}
 & \multirow{3}{*}{\ac{pls}} & No fading & 0.556 & 0.008 & 0.556 & 0.008 \\
 &  & Shallow & 0.593 & 0.073 & 0.457 & 0.004 \\
 &  & Rayleigh & 0.648 & 0.298 & 0.485 & 0.006 \\
\bottomrule[2px]
\end{tabular}
\end{table}

Under no fading, \ac{pld} achieves $\bar{\Tilde{D}}\eve = 0.86$, a $55\%$ gain over \ac{pls} ($0.56$), confirming that the deception mechanism contributes substantially beyond packet erasures alone.

As fading deepens, static designs maintain Eve distortion but Bob's reliability degrades sharply ($\bar{\Tilde{D}}\bob > 0.3$ under Rayleigh) because fixed blocklengths cannot accommodate deep channel dips. The adaptive strategy restores Bob ($\bar{\Tilde{D}}\bob \leq 0.034$) at the cost of reduced Eve distortion: in deep fades Alice must prioritize reliability, leaving fewer degrees of freedom for deception. \Ac{pld} retains its advantage throughout: under adaptive Rayleigh fading, $\bar{\Tilde{D}}\eve \approx 0.55$--$0.57$ versus $\approx 0.49$ for \ac{pls} ($12$--$19\%$ gain).

\section{Conclusion}
\label{sec:conclusion}
We formulated the \ac{pld} design as a Stackelberg game between the transmitter and an adversarial eavesdropper, using semantic distortion as the performance metric. The receiver can adopt one of three decryption strategies (Perception, Dropping, or Exclusion), and we derived closed-form switching surfaces that partition the parameter space into strategy regimes, with an exclusion dominance condition and triple-point analysis characterizing their structure.

The robust operating point $\alpha^{\dagger}$, located at the peak of the worst-case distortion envelope, was shown to be a Stackelberg equilibrium: Alice achieves a provable security guarantee without predicting Eve’s exact strategy. We proved that iterative best-response dynamics oscillate with period two and produce strictly lower time-averaged security than the equilibrium. The closed-form $\alpha^{\dagger}$ thus replaces the iterative algorithm with a simpler and provably better solution.

For the per-strategy resource allocation, we derived closed-form optimal $(n\subscript{M}, n\subscript{K})$ for all feasible (Bob, Eve) strategy pairs. Numerical results confirmed that the regime boundaries match the analytical predictions, that Bob predominantly uses Perception while Eve shifts to Exclusion at small codebook sizes, and that the robust design consistently outperforms the iterative approach across eavesdropper channel conditions.

Under Nakagami-$m$ fading, the static \ac{awgn}-based design preserves eavesdropper distortion but degrades Bob's reliability in deep fading; an adaptive per-block strategy restores reliability at the cost of reduced deception. Across all fading conditions, the \ac{pld} design achieves $12$--$55\%$ higher eavesdropper distortion than a classical erasure-only baseline ($\alpha = 0$). Extending the equilibrium analysis to fading (where strategy boundaries become stochastic) and relaxing the statistical eavesdropper \ac{csi} assumption remain open directions.


%
\bibliographystyle{IEEEtran}
\bibliography{references}


\clearpage

\appendices

\section{Proof of Lemma~\ref{lem:switching_surfaces} (Strategy Switching Surfaces)}
\label{app:proof_switching_surfaces}

Setting $\Delta_1 = \Delta_2$ from \eqref{eq:Tilde_D_i}:
$\varepsilon\subscript{K} \alpha D\subscript{conf} = [\varepsilon\subscript{K} \alpha + (1-\alpha)] D\subscript{loss}$.
Dividing by $D\subscript{conf}$ and rearranging: $\alpha[\varepsilon\subscript{K}(1-\rho) + \rho] = \rho$, yielding \eqref{eq:B12}.
The derivations of \eqref{eq:B13} and \eqref{eq:B23} follow analogously from $\Delta_1 = \Delta_3$ and $\Delta_2 = \Delta_3$, respectively.

\section{Proof of Theorem~\ref{thm:exclusion_dominance} (Exclusion Dominance)}
\label{app:proof_exclusion_dominance}

From \eqref{eq:Tilde_D_i}, the difference $\Delta_2 - \Delta_3$ can be written as:
\begin{equation*}
    \Delta_2 - \Delta_3 = D\subscript{conf}\left[\varepsilon\subscript{K}\alpha\left(\rho - \frac{\mathcal{S}-2}{\mathcal{S}-1}\right) - (1-\alpha)(1-\rho)\right].
\end{equation*}
When $\rho < \frac{\mathcal{S}-2}{\mathcal{S}-1}$, both terms are non-positive for $\alpha \in (0,1)$, so $\Delta_2 < \Delta_3$ universally. The condition $\rho < \frac{\mathcal{S}-2}{\mathcal{S}-1}$ is equivalent to $\mathcal{S} > \frac{2-\rho}{1-\rho}$.

Conversely, when $\mathcal{S} \leq \frac{2-\rho}{1-\rho}$, the switching surface $\mathcal{B}_{23}$ from \eqref{eq:B23} lies in $(0,1)$, implying that Exclusion is optimal for $\alpha > \alpha_{23}^*$.

\section{Proof of Lemma~\ref{lem:triple_point} (Triple Point and Sub-regime Classification)}
\label{app:proof_triple_point}

Setting $\alpha_{12}^* = \alpha_{13}^*$:
\begin{equation*}
    \frac{\rho}{\varepsilon\subscript{K}(1-\rho)+\rho} = \frac{\mathcal{S}-1}{\varepsilon\subscript{K}+\mathcal{S}-1}.
\end{equation*}
Cross-multiplying and simplifying yields $\rho = (\mathcal{S}-1)(1-\rho)$, i.e., $\rho = 1 - 1/\mathcal{S}$.
The same condition ensures $\alpha_{12}^* = \alpha_{23}^*$ and $\alpha_{13}^* = \alpha_{23}^*$, confirming the triple point.
The ordering $\alpha_{12}^* < \alpha_{13}^*$ holds if and only if $\rho < 1 - 1/\mathcal{S}$, as can be verified by direct algebraic comparison of \eqref{eq:B12} and \eqref{eq:B13}.

\section{Proof of Lemma~\ref{lem:impossible_pairs} (Impossible Strategy Pairs)}
\label{app:proof_impossible_pairs}

Each claim follows directly from \eqref{eq:threshold_ordering}. For instance, in~(i): Bob prefers Dropping requires $\alpha > \alpha_{12}^*(\varepsilon\subscript{Bob,K})$, while Eve prefers Perception requires $\alpha < \alpha_{12}^*(\varepsilon\subscript{Eve,K})$. Since $\alpha_{12}^*(\varepsilon\subscript{Eve,K}) < \alpha_{12}^*(\varepsilon\subscript{Bob,K})$, no $\alpha$ satisfies both simultaneously.

\section{Proof of Theorem~\ref{thm:robust_alpha} (Robust Operating Point)}
\label{app:proof_robust_alpha}

Since $\Delta_1 = \varepsilon\subscript{Eve,K} \alpha D\subscript{conf}$ is strictly increasing in $\alpha$, and $\Delta_2$, $\Delta_3$ are both decreasing in $\alpha$, the lower envelope $\min_k \Delta_k(\alpha)$ is a piecewise-linear concave function of $\alpha$. Its maximum is attained at the intersection of $\Delta_1$ with the binding decreasing branch.

\emph{Case 1: $\rho \leq 1 - 1/\mathcal{S}$.} This covers all of Region~A and sub-regime~B1. The active transition is $\Delta_1 \to \Delta_2$ (Perception to Dropping), and the peak of the lower envelope is at $\mathcal{B}_{12}$, yielding $\alpha^{\dagger} = \alpha_{12}^*(\varepsilon\subscript{Eve,K})$. Substituting into $\Delta_1$ gives the corresponding distortion.

\emph{Case 2: $\rho > 1 - 1/\mathcal{S}$.} This is sub-regime~B2, where Dropping is absent. The active transition is $\Delta_1 \to \Delta_3$ (Perception to Exclusion), and the peak is at $\mathcal{B}_{13}$, yielding $\alpha^{\dagger} = \alpha_{13}^*(\varepsilon\subscript{Eve,K})$.

At $\rho = 1 - 1/\mathcal{S}$, both cases give $\alpha^{\dagger} = \frac{\mathcal{S}-1}{\varepsilon\subscript{Eve,K} + \mathcal{S}-1}$, confirming continuity.

\section{Proof of Theorem~\ref{thm:stackelberg} (Stackelberg Equilibrium)}
\label{app:proof_stackelberg}

Part~(i) follows directly from the definition of $\alpha^{\dagger}$: in Case~1 of Theorem~\ref{thm:robust_alpha}, $\alpha^{\dagger} = \alpha_{12}^*$ so $\Delta_1(\alpha^{\dagger}) = \Delta_2(\alpha^{\dagger})$; in Case~2, $\alpha^{\dagger} = \alpha_{13}^*$ so $\Delta_1(\alpha^{\dagger}) = \Delta_3(\alpha^{\dagger})$. Eve's best response is any strategy achieving $\min_k \Delta_k(\alpha^{\dagger})$, which at the switching boundary includes both adjacent strategies.

Part~(ii): The switching threshold $\alpha_{12}^*(\varepsilon\subscript{K}) = \rho / [\varepsilon\subscript{K}(1-\rho) + \rho]$ is strictly decreasing in $\varepsilon\subscript{K}$. Since $\varepsilon\subscript{Bob,K} < \varepsilon\subscript{Eve,K}$ (Bob decodes better than Eve), the threshold ordering \eqref{eq:threshold_ordering} gives $\alpha_{12}^*(\varepsilon\subscript{Eve,K}) < \alpha_{12}^*(\varepsilon\subscript{Bob,K})$. Since $\alpha^{\dagger}$ lies on Eve's switching boundary, $\alpha^{\dagger} = \alpha_{12}^*(\varepsilon\subscript{Eve,K}) < \alpha_{12}^*(\varepsilon\subscript{Bob,K})$, so Bob remains in Perception.

Part~(iii): The lower envelope $\min_k \Delta_k(\alpha)$ is piecewise-linear and concave (Theorem~\ref{thm:robust_alpha} proof). Its unique maximum is at $\alpha^{\dagger}$. Any deviation $\alpha \neq \alpha^{\dagger}$ moves along a descending branch, strictly reducing $\Tilde{D}\eve\superscript{wc}$.

\section{Proof of Theorem~\ref{thm:best_response} (Iterative Best-Response Suboptimality)}
\label{app:proof_best_response}

Part~(i): When Alice assumes Eve adopts Perception, the LP objective is linear and increasing in $\alpha$, so the boundary solution maximizes $\alpha$ subject to Bob's constraint, yielding $\alpha\superscript{o} \gg \alpha_{12}^*(\varepsilon\subscript{Eve,K})$ when Bob's channel is good. Conversely, assuming Eve adopts Dropping or Exclusion, the LP objective is decreasing in $\alpha$, giving $\alpha\superscript{o} = 0$.

Part~(ii): Since $\alpha\superscript{high} > \alpha_{12}^*(\varepsilon\subscript{Eve,K})$, Eve's best response switches away from Perception. Since $\alpha\superscript{low} = 0 < \alpha_{12}^*(\varepsilon\subscript{Eve,K})$, Eve switches back to Perception. The process repeats with period~2.

Part~(iii): By the concavity of $\min_k \Delta_k(\alpha)$, for any $\alpha\superscript{low} < \alpha^{\dagger} < \alpha\superscript{high}$:
\[
    \frac{1}{2}\bigl[\min_k \Delta_k(\alpha\superscript{low}) + \min_k \Delta_k(\alpha\superscript{high})\bigr] < \min_k \Delta_k(\alpha^{\dagger}) = \Delta^{\dagger},
\]
which is Jensen's inequality applied to the concave envelope. The time average inherits this strict inequality.

\section{Proof of Theorem~\ref{thm:D_eve_maximum} (Monotonicity of Eve's Distortion --- Initial Resource Allocation)}
\label{app:proof_D_eve_monotonicity}

We first investigate the monotonicity of $\varepsilon_{i,j}$ with respect to $n_j$. In particular, we have
    \begin{equation}
    \label{eq:app_partial_err_partial_n}
        \frac {\partial {\varepsilon_{i,j}}}{\partial{n_j}}=\frac {\partial {\varepsilon_{i,j}}}{\partial{w_{i,j}}}\frac {\partial {w_{i,j}}}{\partial{n_j}} < 0,
    \end{equation}
    where
    \begin{equation}
        w_{i,j}=\sqrt{\frac{n_j}{V(\gamma_{i})}}\left(\mathcal{C}(\gamma_{i})-\frac{d_j}{n_j}\right)\ln{2},
    \end{equation}
    and
    \begin{equation}
        \frac{\partial \varepsilon_{i,j}}{\partial w_{i,j}}=-\frac{1}{\sqrt{2\pi}}e^{-w_{i,j}^2/2}<0.
    \end{equation}

Since $\mathcal{C}(\gamma_i) > 0$ and $d_j, n_j > 0$, we define $r_j \triangleq d_j / n_j$ and compute
    \begin{equation}
        \frac{\partial w_{i,j}}{\partial n_j} = \frac{\ln 2}{2\sqrt{n_j V(\gamma_i)}} \left(\mathcal{C}(\gamma_i) + r_j\right) > 0.
    \end{equation}
Hence $\varepsilon_{i,j}$ is monotonically decreasing in $n_j$.

Next, we compute the partial derivatives of $D\eve$ with respect to $n\subscript{M}$ and $n\subscript{K}$:
    \begin{align}
    \frac{\partial D\eve}{\partial n\subscript{M}}&=\frac{\partial\varepsilon\subscript{Eve,M}}{\partial n\subscript{M}}(D\subscript{loss}-\alpha\varepsilon\subscript{Eve,K}D\subscript{conf}), \\
    \frac{\partial D\eve}{\partial n\subscript{K}}&=\frac{\partial\varepsilon\subscript{Eve,K}}{\partial n\subscript{K}}\alpha(1-\varepsilon\subscript{Eve,M})D\subscript{conf}< 0.
    \end{align}
    From the second equation, $D\eve$ is monotonically decreasing in $n\subscript{K}$, while the monotonicity of $D\eve$ in $n\subscript{M}$ depends on $(D\subscript{loss}-\alpha\varepsilon\subscript{Eve,K}D\subscript{conf})$. Thus, the optimal $(n\subscript{M},n\subscript{K})$ must be located at the boundary.

    With $n\subscript{K} = n\subscript{K}\superscript{min}$, the upper bound of $\varepsilon\subscript{Bob,K}$ can be expressed as:
    \begin{equation}
        \varepsilon\subscript{Bob,K}\leqslant\varepsilon\subscript{Bob,K}\superscript{max}=\mathcal{Q}\left(\sqrt{\frac{n\subscript{K}\superscript{min}}{V\bob}}\left(\mathcal{C}\bob-\frac{d\subscript{K}}{n\subscript{K}\superscript{min}}\right)\ln 2 \right).
    \end{equation}
    In the above equation, $n\subscript{K}\superscript{min}$ can be derived from the constraint $\varepsilon\subscript{Bob,K}\leqslant \varepsilon\subscript{Bob,K}\superscript{th}$:
    \begin{equation}
        n\subscript{K}\superscript{min}=\phi(\varepsilon\subscript{Bob,K}\superscript{max}).
    \end{equation}

    On the other hand, with $\varepsilon\subscript{Bob,K} = \varepsilon\subscript{Bob,K}\superscript{max}$, the distortion constraint $\Tilde{D}\bob \leq \Tilde{D}\bob\superscript{th}$ gives:
    \begin{equation}
        \varepsilon\subscript{Bob,M} D\subscript{loss} + \alpha(1-\varepsilon\subscript{Bob,M})\varepsilon\subscript{Bob,K}\superscript{max}D\subscript{conf} \leq \Tilde{D}\bob\superscript{th}.
    \end{equation}
    Solving for $\varepsilon\subscript{Bob,M}$:
    \begin{equation}
        \varepsilon\subscript{Bob,M}\leqslant\varepsilon\subscript{Bob,M}\superscript{max}=\frac{\Tilde{D}\bob\superscript{th}-\alpha\varepsilon\subscript{Bob,K}\superscript{max}D\subscript{conf}}{D\subscript{loss}-\alpha\varepsilon\subscript{Bob,K}\superscript{max}D\subscript{conf}}.
    \end{equation}

    With $n\subscript{K}=n\subscript{K}\superscript{min}$ and $\varepsilon\subscript{Eve,K}\superscript{max}$ fixed, the monotonicity of $D\eve$ with respect to $n\subscript{M}$ depends solely on the sign of $D\subscript{loss} - \alpha\varepsilon\subscript{Eve,K}\superscript{max}D\subscript{conf}$:
    \begin{itemize}
        \item When $\alpha \geqslant \frac{D\subscript{loss}}{\varepsilon\subscript{Eve,K}\superscript{max}D\subscript{conf}}$: $D\eve$ is increasing in $n\subscript{M}$, so the maximum is at $(n\subscript{M}\superscript{max}, n\subscript{K}\superscript{min})$.
        \item When $\alpha < \frac{D\subscript{loss}}{\varepsilon\subscript{Eve,K}\superscript{max}D\subscript{conf}}$: $D\eve$ is decreasing in $n\subscript{M}$, so the maximum is at $(n\subscript{M}\superscript{min}, n\subscript{K}\superscript{min})$.
    \end{itemize}

    Combining the two cases, we obtain:
    \begin{equation}
        \max D\eve = \begin{cases}
            D\eve(n\subscript{M}\superscript{max}, n\subscript{K}\superscript{min}), & \alpha \geqslant \frac{D\subscript{loss}}{\varepsilon\subscript{Eve,K}\superscript{max}D\subscript{conf}}, \\
            D\eve(n\subscript{M}\superscript{min}, n\subscript{K}\superscript{min}), & \alpha < \frac{D\subscript{loss}}{\varepsilon\subscript{Eve,K}\superscript{max}D\subscript{conf}}.
        \end{cases}
    \end{equation}

\section{Proof of Theorem~\ref{thm:dropping_monotonicity} (Monotonicity of Eve's Estimated Distortion --- Dropping Strategy)}
\label{app:proof_dropping_monotonicity}

From the error probability monotonicity established in Appendix~\ref{app:proof_D_eve_monotonicity}, we can derive the monotonicity of $\Tilde{D}\eve^{\beta_2}$ with respect to $n\subscript{M}$ and $n\subscript{K}$ respectively:
    \begin{equation}
        \frac{\partial \Tilde{D}\eve^{\beta_2}}{\partial n\subscript{K}}=\frac{\partial\varepsilon\subscript{Eve,K}}{\partial n\subscript{K}}\alpha^{\mathrm{o}}(1-\varepsilon\subscript{Eve,M})D\subscript{loss}<0,
    \end{equation}
    \begin{equation}
        \frac{\partial \Tilde{D}\eve^{\beta_2}}{\partial n\subscript{M}}=\frac{\partial \varepsilon\subscript{Eve,M}}{\partial n\subscript{M}}\alpha^{\mathrm{o}}(1-\varepsilon\subscript{Eve,K})D\subscript{loss}<0.
    \end{equation}

    Thus, $\Tilde{D}\eve^{\beta_2}$ is monotonically decreasing with respect to both $n\subscript{M}$ and $n\subscript{K}$, and the maximum is achieved at $(n\subscript{M}\superscript{min}, n\subscript{K}\superscript{min})$.

\section{Proof of Theorem~\ref{thm:exclusion_monotonicity} (Monotonicity of Eve's Estimated Distortion --- Exclusion Strategy)}
\label{app:proof_exclusion_monotonicity}

We first calculate the first-order derivatives of $\Tilde{D}\eve^{\beta_3}$ with respect to $n\subscript{M}$ and $n\subscript{K}$ respectively:
        \begin{align}
            \begin{split}
                \frac{\partial\Tilde{D}\eve^{\beta_3}}{\partial n\subscript{M}}&=\frac{\partial \varepsilon\subscript{Eve,M}}{\partial n\subscript{M}}\left[D\subscript{loss}-D\subscript{conf}\right. \\
                &\left.+\alpha^{\mathrm{o}} D\subscript{conf}(1-\varepsilon\subscript{Eve,K}\frac{\mathcal{S}-2}{\mathcal{S}-1})\right]
            \end{split}
        \end{align}
        \begin{equation}
            \frac{\partial\Tilde{D}\eve^{\beta_3}}{\partial n\subscript{K}}=\frac{\partial \varepsilon\subscript{Eve,K}}{\partial n\subscript{K}}(1-\varepsilon\subscript{Eve,M})\frac{\alpha\superscript{\mathrm{o}}(\mathcal{S}-2)}{\mathcal{S}-1}D\subscript{conf}<0.
        \end{equation}

    $\Tilde{D}\eve^{\beta_3}$ is monotonically decreasing with respect to $n\subscript{K}$, while the monotonicity with respect to $n\subscript{M}$ depends on $\alpha^{\mathrm{o}}$. When $\alpha^{\mathrm{o}}<\frac{D\subscript{conf}-D\subscript{loss}}{D\subscript{conf}(1-\varepsilon\subscript{Eve,K}\superscript{max}\frac{\mathcal{S}-2}{\mathcal{S}-1})}$, we have $\frac{\partial\Tilde{D}\eve^{\beta_3}}{\partial n\subscript{M}}>0$, so the optimum is at $(n\subscript{M}\superscript{max},n\subscript{K}\superscript{min})$.

    When $\alpha^{\mathrm{o}}\geqslant\frac{D\subscript{conf}-D\subscript{loss}}{D\subscript{conf}(1-\varepsilon\subscript{Eve,K}\superscript{max}\frac{\mathcal{S}-2}{\mathcal{S}-1})}$, we have $\frac{\partial\Tilde{D}\eve^{\beta_3}}{\partial n\subscript{M}}\leqslant 0$, so the optimum is at $(n\subscript{M}\superscript{min},n\subscript{K}\superscript{min})$.

\section{Proof of Proposition~\ref{prop:oscillation} (Oscillation Characterization)}
\label{app:proof_oscillation}

For step~1: when Eve adopts Perception ($\beta_1\superscript{Eve} = 1$) and Bob adopts Perception ($\beta_1\superscript{Bob} = 1$), the ciphering probability optimization gives $\alpha\superscript{o} = \frac{\Tilde{D}\superscript{th}\bob - \varepsilon\subscript{Bob,M}D\subscript{loss}}{(1-\varepsilon\subscript{Bob,M})\varepsilon\subscript{Bob,K}D\subscript{conf}}$. When Bob has a good channel ($\varepsilon\subscript{Bob,K}$ small), this value is large, typically exceeding $\alpha_{12}^*(\varepsilon\subscript{Eve,K})$, which causes Eve to switch strategy.

For step~3: when Eve adopts Dropping or Exclusion and Bob still adopts Perception, the optimization gives $\alpha\superscript{o} = 0$. Since $\alpha_{12}^*(\varepsilon\subscript{Eve,K}) > 0$ for all $\varepsilon\subscript{Eve,K} \in (0,1)$, we have $\alpha\superscript{o} < \alpha_{12}^*(\varepsilon\subscript{Eve,K})$, causing Eve to switch back to Perception.


\end{document}